\def \VersionLong {}
\def \VersionAuthor {}
\ifdefined \VersionLong
	\newcommand{\LongVersion}[1]{\ifdefined\VersionWithComments{\color{red!40!black}#1}\else#1\fi}
	\newcommand{\ShortVersion}[1]{\ifdefined\VersionWithComments{\color{black!40}#1}\fi}
\else
	\newcommand{\LongVersion}[1]{\ifdefined\VersionWithComments{\color{black!40}#1}\fi}
	\newcommand{\ShortVersion}[1]{\ifdefined\VersionWithComments{\color{red!40!black}#1}\else#1\fi}
\fi

\ifdefined \VersionAuthor
	\newcommand{\AuthorVersion}[1]{#1}
	\newcommand{\IEEEVersion}[1]{}
\else
	\newcommand{\AuthorVersion}[1]{}
	\newcommand{\IEEEVersion}[1]{#1}
\fi

\documentclass[10pt, conference, compsocconf]{IEEEtran}
\IEEEoverridecommandlockouts
\usepackage[utf8]{inputenc}
\usepackage[english]{babel}

\usepackage{graphicx}
\graphicspath{{figures/}}

\usepackage[ruled,vlined,linesnumbered]{algorithm2e}
	\SetKwInOut{Input}{input}
	\SetKwInOut{Output}{output}

\usepackage{subcaption}

\usepackage{paralist} % inline lists

\newenvironment{ienumeration}
	{\ifdefined\VersionLong\begin{enumerate}\else\begin{inparaenum}[\itshape i\upshape)]\fi}
	{\ifdefined\VersionLong\end{enumerate}\else\end{inparaenum}\fi}

\usepackage{amsmath} % for \text et al., for ``cases'' environment…
\usepackage{amssymb} % for \mathbb et al.

\ifdefined\VersionWithComments
	\usepackage{draftwatermark}
	\SetWatermarkText{draft}
	\SetWatermarkScale{3}
	\SetWatermarkColor[gray]{0.9}
\fi
\usepackage[svgnames,table]{xcolor}
\definecolor{darkblue}{rgb}{0.0,0.0,0.6}
\definecolor{darkgreen}{rgb}{0, 0.5, 0}
\definecolor{darkpurple}{rgb}{0.7, 0, 0.7}
\definecolor{darkblue}{rgb}{0, 0, 0.7}

\AuthorVersion{
}
	
\usepackage[capitalise,english,nameinlink]{cleveref} % load after algorithm2e and hyperref
\crefname{line}{\text{line}}{\text{lines}} % to remove the capital
\crefname{assumption}{\text{Assumption}}{\text{Assumptions}} % added Assumption

\usepackage{tikz}
\usetikzlibrary{arrows,automata,positioning}
\tikzstyle{every node}=[initial text=]
\tikzstyle{location}=[rectangle, rounded corners, minimum size=12pt, draw=black, fill=blue!10, inner sep=2pt]
\tikzstyle{invariant}=[draw=black, dotted, inner sep=1pt] % xshift=1em, 
\tikzstyle{final}=[double]
\tikzstyle{accepting}=[final]
\tikzstyle{PTPMOPT}=[,dashed,color=red,semithick]

\newcommand{\styleact}[1]{\ensuremath{\textcolor{coloract}{\mathrm{#1}}}}
\newcommand{\styleclock}[1]{\ensuremath{\textcolor{colorclock}{#1}}}

\newcommand{\styleparam}[1]{\ensuremath{\textcolor{colorparam}{#1}}}

\definecolor{coloract}{rgb}{0.50, 0.70, 0.30}
\definecolor{colorclock}{rgb}{0.4, 0.4, 1}
\definecolor{colorconst}{rgb}{0.50, 0.20, 0.00}
\definecolor{colordisc}{rgb}{1, 0, 1}
\definecolor{colorloc}{rgb}{0.4, 0.4, 0.65}
\definecolor{colorparam}{rgb}{1, 0.6, 0.0}

\newcommand{\init}{_0}

\newcommand{\A}{\ensuremath{\mathcal{A}}}
\newcommand{\Actions}{\Sigma}
\newcommand{\action}{\ensuremath{a}}
\newcommand{\ActionsIndices}{\zeta}
\newcommand{\assign}{\leftarrow}
\newcommand{\BranchingCard}{B}

\newcommand{\Clock}{\mathbb{X}} % set of clocks
\newcommand{\ClockCard}{H} % cardinality of clocks
\newcommand{\clock}{x} % clock
\newcommand{\clockabs}{\ensuremath{x_\mathit{abs}}} % clock measuring the absolute time
\newcommand{\clockval}{\nu} % clock valuation
\newcommand{\ClocksZero}{\vec{0}}
\newcommand{\compOp}{\bowtie}
\newcommand{\CTrue}{\mathbf{true}}

\newcommand{\edge}{e}
\newcommand{\Edges}{E}
\newcommand{\longuefleche}[1]{\stackrel{#1}{\longrightarrow}}
\newcommand{\longueflecheRel}[1]{\stackrel{#1}{\mapsto}}

\newcommand{\flecheRel}{{\rightarrow}}

\newcommand{\grandn}{{\mathbb N}}
\newcommand{\grandq}{{\mathbb Q}}
\newcommand{\grandqplus}{\grandq_{+}} % \geq 0

\newcommand{\guard}{g}

\newcommand{\invariant}{I}
\newcommand{\K}{K}

\newcommand{\Lg}{\mathcal{L}}
\newcommand{\loc}{l} % location
\newcommand{\locinit}{\loc\init}
\newcommand{\Loc}{L} % set of locations
\newcommand{\LocFinal}{F}

\newcommand{\Param}{\mathbb{P}} % set of parameters (P)
\newcommand{\param}{p} % parameter (p)
\newcommand{\ParamCard}{M} % number of parameters
\newcommand{\pval}{v} % parameter valuation

\newcommand{\R}{{\mathbb{R}}}
\newcommand{\Rgeqzero}{\R_{\geq 0}}
\newcommand{\Rgzero}{{\R_{>0}}}
\newcommand{\sinit}{s\init} % initial set of states
\newcommand{\somelocs}{T} % subset of locations
\newcommand{\state}{\ensuremath{s}} % concrete state
\newcommand{\States}{S} % for LTS
\newcommand{\word}{\textcolor{colorok}{w}}

\newcommand{\wloc}{w}

\newcommand{\resets}{R}

\newcommand{\reset}[2]{\ensuremath{[#1]_{#2}}}
\newcommand{\valuate}[2]{\ensuremath{#2(#1)}}
\newcommand{\pat}{\mathit{pat}}
\newcommand{\str}{w}
\newcommand{\Zp}{{\mathbb{Z}_{>0}}}
\usepackage{pifont}% http://ctan.org/pkg/pifont
\newcommand{\cmark}{\ding{51}}%
\newcommand{\defProblem}[3]
{%	
\noindent\fcolorbox{black}{blue!15}{
	\begin{minipage}{.95\columnwidth}
		\textbf{#1 problem:}\\
		\textsc{Input}: #2\\
		\textsc{Problem}: #3
	\end{minipage}
}
}

\newcommand{\cellHeader}[1]{\cellcolor{blue!20}\textbf{#1}}
\newcommand{\multiCellHeader}[2]{\multicolumn{#1}{c|}{\cellcolor{blue!20}\textbf{#2}}}
\newcommand{\startMultiCellHeader}[2]{\multicolumn{#1}{|c|}{\cellcolor{blue!20}\textbf{#2}}}

\usepackage{amsthm}
\theoremstyle{plain}
\newtheorem{lemma}{Lemma}

\newtheorem{assumption}{Assumption}

\theoremstyle{definition}
\newtheorem{definition}{Definition}
\newtheorem{example}{Example}

\theoremstyle{remark}
\ifdefined \VersionWithComments
	\usepackage{marginnote}
	\newcommand{\marginX}{\marginnote{\huge{\quad\quad\textbf{!}\quad\quad}}}
	\newcommand{\ea}[1]{\mbox{}{\color{blue}\marginX{}\textbf{[\'Etienne}: #1]}}
	\newcommand{\ih}[1]{\mbox{}{\color{purple}\marginX{}\textbf{[Ichiro}: #1]}}
	\newcommand{\mw}[1]{\mbox{}{\color{orange}\marginX{}\textbf{[Masaki}: #1]}}
	\newcommand{\instructions}[1]{{\color{red}\marginX{}\textbf{[Instructions: ``#1'']}}}
	\newcommand{\reviewer}[2]{\mbox{}{\color{red}\marginX{}\textbf{[Reviewer #1}: ``#2'']}}
	\newcommand{\todo}[1]{\mbox{}{\color{red}{\marginX{}\textbf{TODO}\ifx#1\\\else:\ \fi #1}}} % here, ``\\'' stands for ``empty''
\else
	\newcommand{\instructions}[1]{}
	\newcommand{\ea}[1]{}
	\newcommand{\ih}[1]{}
	\newcommand{\mw}[1]{}
	\newcommand{\reviewer}[2]{}
	\newcommand{\todo}[1]{}
\fi

\ifdefined\VersionWithComments
	\usepackage[pagewise,switch]{lineno} % switch, modulo
	\linenumbers
	
\fi

\newcommand{\stylealgo}[1]{\ensuremath{\textsf{#1}}}
\newcommand{\EFsynth}{\stylealgo{EFsynth}}
\newcommand{\PTPM}{\stylealgo{PTPM}}
\newcommand{\PTPMopt}{\ensuremath{\stylealgo{\PTPM}_{\textsf{opt}}}}
\newcommand{\TransPattern}{\stylealgo{MakeSymbolic}}
\newcommand{\TransWord}{\stylealgo{TW2PTA}}

\newcommand{\imitator}{\textsf{IMITATOR}}

\ifdefined \VersionWithComments
 	\definecolor{colorok}{RGB}{80,80,150}
\else
	\definecolor{colorok}{RGB}{0,0,0}
\fi

\newcommand{\eg}{\textcolor{colorok}{e.\,g.,}\xspace}
\newcommand{\ie}{\textcolor{colorok}{i.\,e.,}\xspace}
\newcommand{\st}{\textcolor{colorok}{s.t.}\xspace}
\newcommand{\viz}{\textcolor{colorok}{viz.,}\xspace}

\title{Offline timed pattern matching under uncertainty\ifdefined \VersionWithComments
	\textcolor{red}{{ [Version: \today{}]}}
\fi\thanks{%
	This work is partially supported by the ANR national research program PACS (ANR-14-CE28-0002),
	by
	JST ERATO HASUO Metamathematics for Systems Design Project (No.\ JPMJER1603),
        and by JSPS Grants-in-Aid No.\ 15KT0012 \& 18J22498.
}}
\author{\IEEEauthorblockN{Étienne André}
\IEEEauthorblockA{\textit{Université Paris 13, LIPN, CNRS, UMR 7030} \\
F-93430, Villetaneuse, France
}
\and
\IEEEauthorblockN{Ichiro Hasuo}
\IEEEauthorblockA{%
\textit{National Institute of Informatics}\\
Sokendai (The Graduate University\\for Advanced Studies)\\
Tokyo, Japan
}
\and
\IEEEauthorblockN{Masaki Waga}
\IEEEauthorblockA{%
\textit{National Institute of Informatics}\\Tokyo, Japan\\
Sokendai (The Graduate University\\for Advanced Studies)\\Kanagawa, Japan\\
JSPS Research Fellow
}
}

\AuthorVersion{
	\newcommand\copyrighttext{%
	\footnotesize
		This is the author version of the manuscript of the same name published in the proceedings of the 23rd International Conference on Engineering of Complex Computer Systems (ICECCS 2018).
	The final version is available at \href{http://www.dx.doi.org/10.1109/ICECCS2018.2018.00010}{10.1109/ICECCS2018.2018.00010}.
	}
	\newcommand\copyrightnotice{%
	\begin{tikzpicture}[remember picture,overlay]
	\node[anchor=south,yshift=10pt] at (current page.south) {\fbox{\parbox{\dimexpr\textwidth-\fboxsep-\fboxrule\relax}{\copyrighttext}}};
	\end{tikzpicture}%
	}
}

\begin{document}

% For all page numbers, except p.1
\pagestyle{plain}

\maketitle
\AuthorVersion{
	\copyrightnotice
}

% For page numbers on p.1
\thispagestyle{plain}

\ifdefined \VersionWithComments
	\textcolor{red}{\textbf{This is the version with comments. To disable comments, comment out line~3 in the \LaTeX{} source.}}
\fi

\begin{abstract}
	Given a log and a specification, timed pattern matching aims at exhibiting for which start and end dates a specification holds on that log.
	For example, ``a given action is always followed by another action before a given deadline''.
	This problem has strong connections with \emph{monitoring} real-time systems.
	We address here timed pattern matching in presence of an \emph{uncertain} specification, \ie{} that may contain timing parameters (\eg{} the deadline can be uncertain or unknown).
	That is, we want to know for which start and end dates, and for what values of the deadline, this property holds.
	Or what is the minimum or maximum deadline (together with the corresponding start and end dates) for which this property holds.
	We propose here a framework for timed pattern matching based on parametric timed model checking.
	In contrast to most parametric timed problems, the solution is effectively computable, and we perform experiments using \imitator{} to show the applicability of our approach.
\end{abstract}


\instructions{Submissions should not exceed 15 pages in length (not including the bibliography, which is thus not restricted), but may be supplemented with a clearly marked appendix, which will be reviewed at the discretion of the program committee.}

\todo{we need to cite \cite{DFM13,AMNU17,BFMU17}}

\ea{perhaps cite later:
\begin{itemize}
	\item Deshmukh, J.V., Majumdar, R., Prabhu, V.S.: Quantifying conformance using the
skorokhod metric. In: Kroening, D., Păsăreanu, C.S. (eds.) CAV 2015. LNCS, vol.
9207, pp. 234–250. Springer, Cham (2015). doi:\url{10.1007/978-3-319-21668-3_14}
	\item Jakšić, S., Bartocci, E., Grosu, R., Ničković, D.: Quantitative monitoring of STL
with edit distance. In: Falcone, Y., Sánchez, C. (eds.) RV 2016. LNCS, vol. 10012,
pp. 201–218. Springer, Cham (2016). doi:\url{10.1007/978-3-319-46982-9_13}
\end{itemize}

}

\ea{hello}
\ih{hello}
\mw{hello}

%%%%%%%%%%%%%%%%%%%%%%%%%%%%%%%%%%%%%%%%%%%%%%%%%%%%%%%%%%%%
%%%%%%%%%%%%%%%%%%%%%%%%%%%%%%%%%%%%%%%%%%%%%%%%%%%%%%%%%%%%
\section{Introduction}\label{section:introduction}
%%%%%%%%%%%%%%%%%%%%%%%%%%%%%%%%%%%%%%%%%%%%%%%%%%%%%%%%%%%%
%%%%%%%%%%%%%%%%%%%%%%%%%%%%%%%%%%%%%%%%%%%%%%%%%%%%%%%%%%%%

Real-time systems are increasingly pervasive in human activities, with systems becoming more and more complex.
Monitoring real-time systems consists in deciding whether a log satisfies a specification.
Often, we are rather interested in knowing \emph{for which segment} of the log the specification is satisfied or violated.
This problem can be related to string matching and pattern matching.
The \emph{timed pattern matching problem} was recently formulated in~\cite{UFAM14}, with subsequent works varying the setting and improving the technique (\eg{} \cite{UFAM16,WAH16,WHS17}).
In~\cite{UFAM14}, the problem takes as input a timed signal $\word$ (values that change over the continuous notion of time) and a timed regular expression $R$ (a real-time extension of regular expressions); and it returns the set of intervals $[t, t']$ for which the timed regular expression is matched by the log, \ie{} $\word$ restricted to $[t,t']$ belongs to the language of~$R$.

In~\cite{WAH16,WHS17}, we introduced a solution to the timed pattern matching problem where the log is given in the form of a timed word (a sequence of events with their associated timestamps), and the specification in the form of a timed automaton (TA), an extension of finite-state automata with clocks~\cite{AD94}.
In~\cite{WAH16}, our technique first relied on an extension of the Boyer-Moore algorithm for string matching and its extension to (untimed) pattern matching by Watson and Watson;
in~\cite{WHS17}, it relied on (an automata-theoretic extension of) skip values from the Franek--Jennings--Smyth (FJS) algorithm for string matching~\cite{FJS07}, so as to improve the efficiency.

As a motivating example, consider the example in \cref{figure:example}.
Here \$ is a special terminal character.
Consider the automaton in \cref{figure:example:PTA}, and fix $\param_1 = 1$ and $\param_2 = 1$---which gives a timed automaton~\cite{AD94}.
For this timed automaton (say~$\A$) and the target timed word $\word$ in \cref{figure:example:word}, the output of the timed pattern matching problem is the set of matching intervals $\{(t,t') \mid \word|_{(t,t')} \in \Lg(\A) \} = \{(t,t') \mid t \in (3.7,3.9), t' \in [6.0,\infty)\}$.

%%%%%%%%%%%%%%%%%%%%%%%%%%%%%%%%%%%%%%%%%%%%%%%%%%%%%%%%%%%%
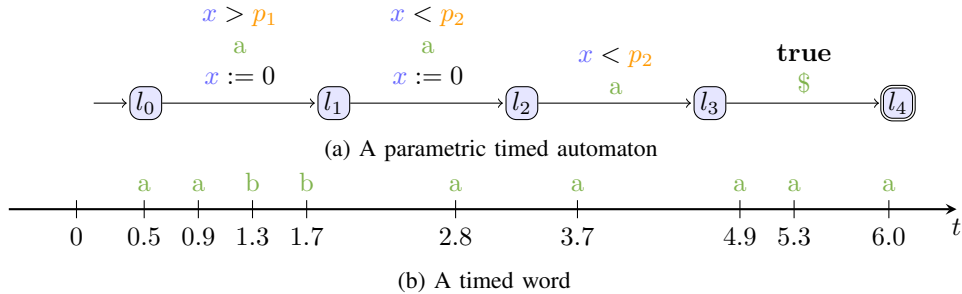
\begin{figure*}[t]
	\begin{subfigure}[b]{\textwidth}
	\centering
		\scalebox{1}{
		\begin{tikzpicture}[shorten >=1pt,node distance=2.5cm,on grid,auto]
		%% states
		\node[location,initial] (s_0) {$\loc_0$};
		\node[location] (s_1) [right of=s_0] {$\loc_1$};
		\node[location] (s_2) [right of=s_1] {$\loc_2$};
		\node[location] (s_3) [right of=s_2]{$\loc_3$};
		\node[location,accepting] (s_4) [right of=s_3]{$\loc_4$};

		%% edges
		\path[->] 
		(s_0) edge [above] node {\begin{tabular}{c}
									$\styleclock{x} > \styleparam{\param_1}$\\
									$\styleact{a}$\\
									$\styleclock{x} := 0$
								\end{tabular}} (s_1)
		(s_1) edge [above] node {\begin{tabular}{c}
									$\styleclock{x} < \styleparam{\styleparam{\param_2}}$\\
									$\styleact{a}$\\
									$\styleclock{x} := 0$
								\end{tabular}} (s_2)
		(s_2) edge [above] node[align=center] {$\styleclock{x} < \styleparam{\styleparam{\param_2}}$\\$\styleact{a}$} (s_3)
		(s_3) edge [above] node[align=center] {$\CTrue$\\$\styleact{\$}$} (s_4);
		\end{tikzpicture}}
	\caption{A parametric timed automaton}
	\label{figure:example:PTA}
	\end{subfigure}
	
	%------------------------------------------------------------
	\begin{subfigure}[b]{0.99\textwidth}
	\centering
	\scalebox{1}{
	\begin{tikzpicture}[scale=1.2,xscale=1.5]
	% scale
	\draw [thick, -stealth](-0.5,0)--(6.5,0) node [anchor=north]{$t$};
	\draw (0,0.1) -- (0,-0.1) node [anchor=north]{$0$};

	% alphabets
	\draw (0.5,0.1) node[anchor=south]{$\styleact{a}$} -- (0.5,-0.1) node[anchor=north]{$0.5$};
	\draw (0.9,0.1) node[anchor=south]{$\styleact{a}$} -- (0.9,-0.1) node[anchor=north]{$0.9$};
	\draw (1.3,0.1) node[anchor=south]{$\styleact{b}$} -- (1.3,-0.1) node[anchor=north]{$1.3$};
	\draw (1.7,0.1) node[anchor=south]{$\styleact{b}$} -- (1.7,-0.1) node[anchor=north]{$1.7$};
	\draw (2.8,0.1) node[anchor=south]{$\styleact{a}$} -- (2.8,-0.1) node[anchor=north]{$2.8$};
	\draw (3.7,0.1) node[anchor=south]{$\styleact{a}$} -- (3.7,-0.1) node[anchor=north]{$3.7$};
	\draw (5.3,0.1) node[anchor=south]{$\styleact{a}$} -- (5.3,-0.1) node[anchor=north]{$5.3$};
	\draw (4.9,0.1) node[anchor=south]{$\styleact{a}$} -- (4.9,-0.1) node[anchor=north]{$4.9$};
	\draw (6.0,0.1) node[anchor=south]{$\styleact{a}$} -- (6.0,-0.1) node[anchor=north]{$6.0$};
	\end{tikzpicture}}
	\caption{A timed word}
	\label{figure:example:word}
	\end{subfigure}

	\caption{An example of parametric timed pattern matching}
	\label{figure:example}
\end{figure*}
%%%%%%%%%%%%%%%%%%%%%%%%%%%%%%%%%%%%%%%%%%%%%%%%%%%%%%%%%%%%

\paragraph{Contribution}
In this work, we consider a more abstract problem:
given a (concrete) timed log and an incomplete specification where some of the timing constants may be known with limited precision or completely unknown, what are the intervals and the valuations of the parameters for which the specification holds?
Coming back to \cref{figure:example}, the question becomes to exhibit values for $t, t', \param_1, \param_2$ for which the specification holds on the log, \ie{}
	$\{(t,t',\pval)\mid \word|_{(t,t')}\in \Lg(\valuate{\A}{\pval})\}$, where $\pval$ denotes a valuation of $\param_1, \param_2$ and $\valuate{\A}{\pval}$ denotes the replacement of $\param_1, \param_2$ in~$\A$ with their respective valuation in~$\pval$.

We introduce an approach using as underlying formalisms timed words and parametric timed automata~\cite{AHV93}, two well-known formalisms in the real-time systems community.
We first show that the problem is decidable (which mainly comes from the fact that logs are finite), and we propose a practical solution based on parametric timed model checking.
\LongVersion{

}%
We implement our method using \imitator{}~\cite{AFKS12} and we perform a set of experiments on a set of automotive benchmarks.

\paragraph{Related work}
Several algorithms have been proposed for online monitoring of real-time temporal logic specifications.
Online monitoring consists in monitoring \LongVersion{on-the-fly }at runtime, while offline monitoring can possibly be performed after the execution is completed, with less hard constraints on the monitoring algorithm performance.
An online monitoring algorithm for ptMTL (a past time fragment of MTL~\cite{Koymans90}) was proposed in~\cite{RFB14} and an algorithm for MTL[U,S] (a variant of MTL with both forward and backward temporal modalities) was proposed in~\cite{HOW14}.
In addition, a case study on an autonomous research vehicle monitoring~\cite{KCDK15} shows such procedures can be performed in an actual vehicle.
% 	\ih{---this is where our motivation comes from, too}.

% BEGIN: removed from WHS17
% We have chosen timed automata as a specification formalism. This is because of their expressiveness as well as various techniques that operate on them.
% Some other formalisms can be translated to timed automata, and via translation,
% % to Alur-Dill style timed automata~\cite{AD94}
% our algorithm offers to these formalisms 
%  an online monitoring algorithm.
% In~\cite{Asarin2002}, a variant of \emph{timed regular expressions (TREs)} are
% proved to have the same expressive power as timed automata.
% For MTL and MITL, transformations into automata are introduced for many 
% different settings; see
% \eg{} \cite{DBLP:conf/focs/AlurH92,DBLP:conf/formats/MalerNP06,DBLP:conf/formats/NickovicP10,DBLP:conf/formats/KiniKP11,d2013clock}.
% END: removed from WHS17

The approaches most related to ours are~\cite{UFAM14,UFAM16,Ulus17}.
In that series of works, logs are encoded by \emph{signals}, \ie{} values that vary over time.
This can be seen as a \emph{state-based} view, while our timed words are \emph{event-based}.
The formalism used for specification in~\cite{UFAM14,UFAM16} is timed regular expressions (TREs).
An offline monitoring algorithm is presented in~\cite{UFAM14} and an online one is in~\cite{UFAM16}.
These algorithms are implemented in the tool \emph{Montre}~\cite{Ulus17}.
% 	, with which we conduct performance comparison.
\LongVersion{%
	The difference between different specification formalisms (TREs, timed automata, temporal logics, etc.) are subtle, but for many realistic examples the difference may not matter.
}
% BEGIN: removed from WHS17
% In the current paper we exploit various operations on automata, most notably zone-based abstraction.
% END: removed from WHS17

% NOTE: former contributions by Waga et al.:
In \cite{WHS17}, we presented an efficient algorithm for online timed pattern matching that employs (an automata-theoretic extension of) skip values from the Franek--Jennings--Smyth (FJS)  algorithm for string matching~\cite{FJS07}.
We showed that our algorithm generally outperforms a brute-force one and our previous BM algorithm~\cite{WAH16}: it is twice as fast for some realistic automotive examples.
Through our theoretical analysis as well as experiments on memory consumption, our algorithm was shown to be suited for online usage scenarios, too.
% We also compare its performance with a recent tool \emph{Montre} for timed pattern matching~\cite{Ulus17}, and observe that ours is faster, at least in terms of the  implementations currently available. 
In comparison, we use here a different approach not based on FJS, but rather on parametric model checking.
%

% BEGIN: removed because we found such a previous work
% None of the previous approaches allow uncertainty in the specification.
% END: removed because we found such a previous work

% BEGIN FORMATS 2017, useless here
% In its course we have obtained an FJS-type algorithm for \emph{untimed} pattern matching, which is one of the  main contributions too. The algorithm  is explained rather in detail, so that it paves the way to our FJS-type \emph{timed} pattern matching that is more complex. 
% END FORMATS 2017, useless here

% BEGIN: removed because model checking is used in \cite{FR08}
% While the use of model checking to solve pattern matching is not new (\eg{} \cite{RG01}), this is, to the best of our knowledge, the first use of parametric model checking techniques to that purpose.
% END: removed because model checking is used in \cite{FR08}

Some algorithms have also been proposed for parameter identification of a temporal logic specification with uncertainty over a log.
For discrete time setting, an algorithm for an extension of LTL is proposed in~\cite{FR08} and for real-time setting, algorithms for parametric signal temporal logic (PSTL) are proposed in~\cite{ADMN11,BFM18}.
Although these works are related to our approach, previous approaches do not focus on segments of a log but one whole log.
In contrast, we exhibit intervals together with their associated parameter valuations, in a fully symbolic fashion.
We believe our matching-based setting is advantageous in many usage scenarios \eg{} from hours of a log of a car, extracting timing constraints of a certain actions to cause slipping.
Also, our setting allows the patterns with complex timing constraints (see the pattern in \cref{figure:patterns:blowup} for example).
% \cref{ss:XP:blowup}

In~\cite{BFMU17}, the robust pattern matching problem is considered over signal regular expressions, consisting in computing the \emph{quantitative} (robust) semantics of a signal relative to an expression.
For piecewise-constant and piecewise-linear signals, the problem can be effectively solved using a finite union of zones.

\LongVersion{Finally, in an orthogonal direction, in~\cite{CJL17}, the authors propose a refinement of trace abstractions technique to perform parameter synthesis for real-time systems; they use \imitator{}, as in our work.}

%%%%%%%%%%%%%%%%%%%%%%%%%%%%%%%%%%%%%%%%%%%%%%%%%%%%%%%%%%%%
\begin{table*}[t]
	\centering
	
	\scalebox{1}{%
	\begin{tabular}{c||c|c|c}
		&log, target&specification, pattern &output\\\hline\hline
		string matching&
		a word $\str\in \Sigma^{*}$
		&
		a word $\pat\in\Sigma^{*}$
		&
		$\{(i,j)\in(\Zp)^{2}\mid \str(i,j)=\pat\}$
		%------------------------------------------------------------
		\\\hline
		pattern matching (PM) &
		a word $\str\in \Sigma^{*}$
		&
		an NFA $\A$
		&
		$\{(i,j)\in(\Zp)^{2}\mid \str(i,j)\in \Lg(\A)\}$
		%------------------------------------------------------------
		\\\hline
		timed PM&
		a timed word $\str\in(\Sigma \times \Rgzero)^{*}$
		&
		a TA $\A$
		&
		$\{(t,t') \in (\Rgzero)^{2}\mid \word|_{(t,t')} \in \Lg(\A)\}$
		\\\hline
		%------------------------------------------------------------
		parametric timed PM&
		a timed word $\str\in(\Sigma \times \Rgzero)^{*}$
		&
		a PTA $\A$
		&
		$\{(t,t', \pval) \mid \word|_{(t,t')} \in \Lg(\valuate{\A}{\pval})\}$
		\\\hline
		%------------------------------------------------------------
	\end{tabular}
	}
	
	\caption{Matching problems}
	\label{table:matching}
\end{table*}
%%%%%%%%%%%%%%%%%%%%%%%%%%%%%%%%%%%%%%%%%%%%%%%%%%%%%%%%%%%%

A summary of various matching problems is given in \cref{table:matching}.

% BEGIN FORMATS 2017, useless here
% A central theme of the paper is benefits of the formalism of \emph{automata}, a mathematical tool whose use is nowadays widespread in fields like temporal logic, model checking, and so on.
% We follow Watson \& Watson's idea of extending skipping from string matching to pattern matching~\cite{Watson2003}, where the key is overapproximation of words and languages by states of automata.
% Our main contribution on the conceptual side is that the same idea applies to \emph{timed} automata as well, where we rely on \emph{zone}-based abstraction (see \eg{} \cite{DBLP:conf/tacas/BehrmannBFL03,DBLP:journals/sttt/BehrmannBLP06,DBLP:conf/cav/HerbreteauSW10}) for computing reachability.
% END FORMATS 2017, useless here

\paragraph{Outline}
We first introduce the parametric timed pattern matching problem and the necessary preliminaries in \cref{section:preliminaries}.
We then introduce our approach based on parametric model checking in \cref{section:approach}.
We apply it to benchmarks from~\cite{HAF14} in \cref{section:experiments}.
We \LongVersion{finally }outline future directions of research in \cref{section:conclusion}.

%%%%%%%%%%%%%%%%%%%%%%%%%%%%%%%%%%%%%%%%%%%%%%%%%%%%%%%%%%%%
%%%%%%%%%%%%%%%%%%%%%%%%%%%%%%%%%%%%%%%%%%%%%%%%%%%%%%%%%%%%
\section{Preliminaries and objective}\label{section:preliminaries}
%%%%%%%%%%%%%%%%%%%%%%%%%%%%%%%%%%%%%%%%%%%%%%%%%%%%%%%%%%%%
%%%%%%%%%%%%%%%%%%%%%%%%%%%%%%%%%%%%%%%%%%%%%%%%%%%%%%%%%%%%

Our target strings are \emph{timed words}~\cite{AD94}, that are time-stamped words over an alphabet $\Actions$.
Our patterns are given by parametric timed automata~\cite{AHV93}.

% As the \emph{target string} of timed pattern matching, we employ
% \emph{timed words}~\cite{AD94}.
% A timed word is a sequence of events with timestamps.
% In this paper, we assume the alphabet $\Actions$ of a timed word does not
% contain the empty character $\varepsilon$ for the comvinience as a
% target string.
% In timed pattern matching, a trimmed segment of a timed word on an
% open interval is considered, that is a 
% \emph{timed word segment}~\cite{DBLP:conf/formats/WagaAH16}.
% Formally, a timed word and a timed word segment are as follows.

\LongVersion{
%%%%%%%%%%%%%%%%%%%%%%%%%%%%%%%%%%%%%%%%%%%%%%%%%%%%%%%%%%%%
\subsection{Timed words and timed segments}
%%%%%%%%%%%%%%%%%%%%%%%%%%%%%%%%%%%%%%%%%%%%%%%%%%%%%%%%%%%%
}

% \begin{definition}[timed word, timed word segment]
For an alphabet $\Actions$, a \emph{timed word} is a
sequence $\word$ of pairs $(\action_i,\tau_i) \in (\Sigma \times \Rgeqzero)$
satisfying $\tau_i < \tau_{i + 1}$ for any $i \in [1,|\word|-1]$.
% 	\ea{the problem is that this definition does not allow for actions in 0-time… But if we replace $\Rgzero$ with $\Rgeqzero$, then the notion of projection $\word|_{(t,t')}$ gets tricky as $\tau_{i-1} < t \leq \tau_i$ may be impossible if $\tau_{i-1} = \tau_i$. So, so far, I keep $\Rgzero$, but if you have a better idea, it'd be welcome.}
Given a timed word $\word$, we often denote it by $(\overline{a},\overline{\tau})$, where $\overline{a}$ is the sequence $(\action_1, \action_2, \cdots)$ and $\overline{\tau}$ is the sequence $(\tau_1, \tau_2, \cdots)$.
Let $\word = (\overline{a},\overline{\tau})$ be a timed word.
We denote the subsequence $(\action_i, \tau_i),(\action_{i+1},
\tau_{i+1}),\cdots,(\action_j,\tau_j)$ by $\word (i,j)$.
For $t \in \R$ such that $- \tau_1 < t$, the \emph{$t$-shift} of $\word$ is
$(\overline{a}, \overline{\tau}) + t = (\overline{a}, \overline{\tau} +
t)$ where
$\overline{\tau} + t = \tau_1 + t,\tau_2 + t,\cdots, \tau_{|\tau|} + t$.
For timed words $\word = (\overline{a},\overline{\tau})$ and
$\word' = (\overline{a'},\overline{\tau'})$,
their \emph{absorbing concatenation} is $\word \circ \word' = (\overline{a} \circ \overline{a'}, \overline{\tau} \circ \overline{\tau'})$ where $\overline{a} \circ \overline{a'}$ and  $\overline{\tau} \circ \overline{\tau'}$ are usual concatenations%, and their \emph{non-absorbing concatenation} is $\word \cdot \word' = \word \circ (\word' + \tau_{|\word|})$
.
% We note that the absorbing concatenation $\word \circ \word'$ is defined only when $\tau_{|\word|} < \tau'_{1}$.

For a timed word $\word = (\overline{a}, \overline{\tau})$ on $\Actions$
and 
$t,t' \in \Rgeqzero$ satisfying $t < t'$, a \emph{timed word segment}
$\word|_{(t,t')}$ is defined by the timed word $(\word (i,j) - t) \circ (\$,t' - t)$
% 	\ea{in \cite{WHS17}, you wrote $(\word (i,j) - t) \circ (\$,t')$; was that a typo? I think it was (and so does reviewer 2) but please confirm} % NOTE: Masaki has confirmed
on the augmented alphabet $\Sigma \sqcup \{\$\}$,
where $i,j$ are chosen so that
$\tau_{i-1} < t \leq \tau_i$ and
$\tau_{j} \leq t' < \tau_{j+1}$.\footnote{%
	Observe that, in contrast with \cite{WAH16,WHS17}, we use here \emph{closed} intervals, \ie{} $t,t'$ may be equal to the start or end time of the word.
}\ea{Just realized $(\word (i,j) - t)$ was not defined… (we only defined $+t$ with $t \in \Rgzero$) I try to fix it using reviewer 2's comments}
Here the fresh symbol \(\$\) is called the \emph{terminal character}.
% \end{definition}

% \ea{in fact, I realize that the counterexample shown by reviewer 2 is still valid :( I'll try to figure out something soon; I guess closed bounds + no events in 0-time + patching slightly our construction might solve it.}

\LongVersion{
%%%%%%%%%%%%%%%%%%%%%%%%%%%%%%%%%%%%%%%%%%%%%%%%%%%%%%%%%%%%
\subsection{Clocks, parameters and guards}
%%%%%%%%%%%%%%%%%%%%%%%%%%%%%%%%%%%%%%%%%%%%%%%%%%%%%%%%%%%%
}

%Let $\interval(\grandn)$ denote the set of non-necessarily closed intervals on~$\grandn$, \ie{} of the form $[a,b]$, $(a,b]$, $[a,b)$ or $(a,b)$ where $a,b\in \grandn$ and $a \leq b$.

We assume a set~$\Clock = \{ \clock_1, \dots, \clock_\ClockCard \} $ of \emph{clocks}, \ie{} real-valued variables that evolve at the same rate.
A clock valuation is\LongVersion{ a function}
$\clockval : \Clock \rightarrow \Rgeqzero$.
% \LongVersion{We identify a clock valuation~$\clockval$ with the point $(\clockval(\clock_1), \dots, \clockval(\clock_{\ClockCard}))$ of $\Rgeqzero^\ClockCard$.
% }
% An integer clock valuation is a valuation $\clockval : \Clock \rightarrow \grandn$.
% We denote by $\Clock = 0$ the conjunction of equalities that assigns 0 to all clocks in~$\Clock$.
% \dl{OK, j'ai juste changé la notation vu que ça apparaît à plusieurs endroits en fait. J'ai changé partout bien sûr.}\ea{sans mettre de macro ;-)}\dl{certes, mais vec n'est utilisé que pour ça, c'est pas trop dur de les retrouver}\ea{ben oui mais moi j'en ai mis une dès mon premier commentaire (nan mais oh)}
We write $\ClocksZero$ for the clock valuation assigning $0$ to all clocks.
Given $d \in \Rgeqzero$, $\clockval + d$ \ShortVersion{is}\LongVersion{denotes the valuation} \st{} $(\clockval + d)(\clock) = \clockval(\clock) + d$, for all $\clock \in \Clock$.
Given $\resets \subseteq \Clock$, we define the \emph{reset} of a valuation~$\clockval$, denoted by $\reset{\clockval}{\resets}$, as follows: $\reset{\clockval}{\resets}(\clock) = 0$ if $\clock \in \resets$, and $\reset{\clockval}{\resets}(\clock)=\clockval(\clock)$ otherwise.

% Throughout this paper,
We assume a set~$\Param = \{ \param_1, \dots, \param_\ParamCard \} $ of \emph{parameters}\LongVersion{, \ie{} unknown constants}.
A parameter {\em valuation} $\pval$ is\LongVersion{ a function}
$\pval : \Param \rightarrow \grandqplus$.\footnote{%
	We choose $\grandqplus$ by consistency with most of the PTA literature, but also because, for classical PTAs, choosing~$\Rgeqzero$ leads to undecidability~\cite{Miller00}.
}
%\st{} $\pval(d)=d$ if~$d\in\grandn$.
% \LongVersion{We identify a valuation~$\pval$ with the point $(\pval(\param_1), \dots, \pval(\param_{\ParamCard}))$ of $\grandqplus^\ParamCard$.
% An \emph{integer} parameter valuation is such that $\forall \param\in \Param, \pval(\param)\in \grandn$. % : \Param \rightarrow \grandn$.
% A clock is \emph{parametric} if it is compared at least once to a parameter, and \emph{non-parametric} otherwise.
% 
% }
We assume ${\compOp} \in \{<, \leq, =, \geq, >\}$.
A guard~$\guard$ is a constraint over $\Clock \cup \Param$ defined by a conjunction of inequalities of the form $\clock \compOp d$, 
or $\clock \compOp \param$ with $d \in \grandn$ and $\param \in \Param$.
Given~$\guard$, we write~$\clockval\models\pval(\guard)$ if %$\valuate{\valuate{\guard}{\pval}}{\clockval}$
the expression obtained by replacing each~$\clock$ with~$\clockval(\clock)$ and each~$\param$ with~$\pval(\param)$ in~$\guard$ evaluates to true.

% Given~$\guard$, $\valuate{\valuate{\guard}{\pval}}{\clockval}$ is the expression obtained by replacing 
% each~$\clock$ with~$\clockval(\clock)$ and each~$\param$ with~$\pval(\param)$.
% We write~$\clockval\models\pval(\guard)$ if~$\valuate{\valuate{\guard}{\pval}}{\clockval}$ evaluates to true.

% 	$\clock \compOp z$, where $z$ is either a parameter or a constant in~$\grandz$.\ea{WARNING: nos preuves utilisent des expressions $y = a + 1$}
% 	$\clock \compOp \sum_{1 \leq j \leq \ParamCard} \beta_j \param_j + d$, with $\beta_j \in \{0, 1 \}$ and $d \in \grandz$.
% We denote the set of all guards by~$\Guards$.
% 

%%%%%%%%%%%%%%%%%%%%%%%%%%%%%%%%%%%%%%%%%%%%%%%%%%%%%%%%%%%%
\subsection{Parametric timed automata}
%%%%%%%%%%%%%%%%%%%%%%%%%%%%%%%%%%%%%%%%%%%%%%%%%%%%%%%%%%%%

Parametric timed automata (PTA) extend timed automata with parameters within guards and invariants in place of integer constants~\cite{AHV93}.

%----------------------------------------------------------
\begin{definition}[PTA]\label{def:uPTA}
	A PTA $\A$ is a tuple \mbox{$\A = (\Actions, \Loc, \locinit, \LocFinal, \Clock, \Param, \invariant, \Edges)$}, where:
	\begin{enumerate}
		\item $\Actions$ is a finite set of actions,
		\item $\Loc$ is a finite set of locations,
		\item $\locinit \in \Loc$ is the initial location,
		\item $\LocFinal \subseteq \Loc$ is the set of accepting locations,
		\item $\Clock$ is a finite set of clocks,
		\item $\Param$ is a finite set of parameters,
		\item $\invariant$ is the invariant, assigning to every $\loc\in \Loc$ a guard $\invariant(\loc)$,
		\item $\Edges$ is a finite set of edges  $\edge = (\loc,\guard,\action,\resets,\loc')$
		where~$\loc,\loc'\in \Loc$ are the source and target locations, $\action \in \Actions$, $\resets\subseteq \Clock$ is a set of clocks to be reset, and $\guard$ is a guard.
	\end{enumerate}
\end{definition}
%----------------------------------------------------------

Given\LongVersion{ a parameter valuation}~$\pval$, we denote by $\valuate{\A}{\pval}$ the non-parametric structure where all occurrences of a parameter~$\param_i$ have been replaced by~$\pval(\param_i)$.
We denote as a \emph{timed automaton} any structure $\valuate{\A}{\pval}$, by assuming a rescaling of the constants: by multiplying all constants in $\valuate{\A}{\pval}$  by their least common denominator, we obtain an equivalent (integer-valued) TA\LongVersion{, as defined in \cite{AD94}}.

The synchronous product (using strong broadcast, \ie{} synchronization on shared actions) of several PTAs gives a PTA.

%------------------------------------------------------------
\begin{definition}[synchronized product of PTAs]
	Let $N \in \grandn$.
	Given a set of PTAs $\A_i = (\Actions_i, \Loc_i, (\locinit)_i, \LocFinal_i, \Clock_i, \Param_i, \invariant_i, \Edges_i)$, $1 \leq i \leq N$,
	the \emph{synchronized product} of $\A_i$, $1 \leq i \leq N$,
	denoted by $\A_1 \parallel \A_2 \parallel \cdots \parallel \A_N$,
	is the tuple
		$(\Actions, \Loc, \locinit, \LocFinal, \Clock, \Param, \invariant, \Edges)$, where:
	\begin{enumerate}
		\item $\Actions = \bigcup_{i=1}^N\Actions_i$,
		\item $\Loc = \prod_{i=1}^N \Loc_i$,
		\item $\locinit = ((\locinit)_1, \dots, (\locinit)_N)$,
		\item $\LocFinal = \{(\loc_1, \dots, \loc_N) \in \Loc \mid \exists i \in [1, N]$ s.t.\ $\loc_i \in \LocFinal_i \}$,
		\item $\Clock = \bigcup_{1 \leq i \leq N} \Clock_i$,
		\item $\Param = \bigcup_{1 \leq i \leq N} \Param_i$,
		\item $\invariant((\loc_1, \dots, \loc_N)) = \bigwedge_{i = 1}^{N} \invariant_i(\loc_i)$ for all $(\loc_1, \dots, \loc_N) \in \Loc$, 
	\end{enumerate}
	and $\Edges{}$ is defined as follows.
	For all $\action \in \Actions$,
	let $\ActionsIndices_\action$ be the subset of indices $i \in 1, \dots, N$
	such that $\action \in \Actions_i$.
	For all  $\action \in \Actions$,
	for all $(\loc_1, \dots, \loc_N) \in \Loc$,
	for all \mbox{$(\loc_1', \dots, \loc_N') \in \Loc$},
	$\big((\loc_1, \dots, \loc_N), \guard, \action, \resets, (\loc'_1, \dots, \loc'_N)\big) \in \Edges$
	if:
	\begin{itemize}
		\item for all $i \in \ActionsIndices_\action$, there exist $\guard_i, \resets_i$ such that $(\loc_i, \guard_i, \action, \resets_i, \loc_i') \in \Edges_i$, $\guard = \bigwedge_{i \in \ActionsIndices_\action} \guard_i$, $\resets = \bigcup_{i \in \ActionsIndices_\action}\resets_i$, and,
		\item for all $i \not\in \ActionsIndices_\action$, $\loc_i' = \loc_i$.
\end{itemize}
\end{definition}
%------------------------------------------------------------

Let us now recall the concrete semantics of TA.

%----------------------------------------------------------
\begin{definition}[Semantics of a TA]
	Given a PTA $\A = (\Actions, \Loc, \locinit, \LocFinal, \Clock, \Param, \invariant, \Edges)$,
	and a parameter valuation~\(\pval\),
	the semantics of $\valuate{\A}{\pval}$ is given by the timed transition system (TTS) $(\States, \sinit, \flecheRel)$, with
	\begin{itemize}
		\item $\States = \{ (\loc, \clockval) \in \Loc \times \Rgeqzero^\ClockCard \mid \clockval \models \valuate{\invariant(\loc)}{\pval} \}$, % \valuate{\valuate{\invariant(\loc)}{\pval}}{\clockval} \text{ evaluates to true} 
		\item $\sinit = (\locinit, \ClocksZero) $,
		\item  $\flecheRel$ consists of the discrete and (continuous) delay transition relations:
		\begin{ienumeration}
			\item discrete transitions: $(\loc,\clockval) \longueflecheRel{\edge} (\loc',\clockval')$, %with $\action \in \Sigma$,
				if $(\loc, \clockval) , (\loc',\clockval') \in \States$, and there exists $\edge = (\loc,\guard,\action,\resets,\loc') \in \Edges$, such that $\clockval'= \reset{\clockval}{\resets}$, and $\clockval\models\pval(\guard$).
				% $\valuate{\valuate{\guard}{\pval}}{\clockval}$ evaluates to true.
			\item delay transitions: $(\loc,\clockval) \longueflecheRel{d} (\loc, \clockval+d)$, with $d \in \Rgeqzero$, if $\forall d' \in [0, d], (\loc, \clockval+d') \in \States$.
		\end{ienumeration}
	\end{itemize}
\end{definition}

    Moreover we write $(\loc, \clockval)\longuefleche{(\edge, d)} (\loc',\clockval')$ for a combination of a delay and discrete transition if
		% d,
		$\exists  \clockval'' :  (\loc,\clockval) \longueflecheRel{d} (\loc,\clockval'') \longueflecheRel{\edge} (\loc',\clockval')$.

Given a TA~$\valuate{\A}{\pval}$ with concrete semantics $(\States, \sinit, \flecheRel)$, we refer to the states of~$\States$ as the \emph{concrete states} of~$\valuate{\A}{\pval}$.
A \emph{run} of~$\valuate{\A}{\pval}$ is a possibly infinite\ea{not necessary: fix it later} alternating sequence of concrete states of $\valuate{\A}{\pval}$ and pairs of edges and delays starting from the initial state $\sinit$ of the form
$\state_0, (\edge_0, d_0), \state_1, \cdots$
with
$i = 0, 1, \dots$, $\edge_i \in \Edges$, $d_i \in \Rgeqzero$ and $(\state_i , \edge_i , \state_{i+1}) \in \flecheRel$.
Given such a run, the associated \emph{timed word} is $(\action_1, \tau_1), (\action_2, \tau_2), \cdots$, where $\action_i$ is the action of edge~$\edge_{i-1}$, and $\tau_i = \sum_{0 \leq j \leq i-1} d_j$, for $i = 1, 2 \cdots$.\footnote{%
	The ``$-1$'' in indices comes from the fact that, following usual conventions in the literature, states are numbered starting from~0 while words are numbered from~1.
}
Given\LongVersion{ a state}~$\state=(\loc, \clockval)$, we say that $\state$ is reachable in~$\valuate{\A}{\pval}$ if $\state$ appears in a run of $\valuate{\A}{\pval}$.
By extension, we say that $\loc$ is reachable; and by extension again, given a set~$\somelocs$ of locations, we say that $\somelocs$ is reachable if there exists $\loc \in \somelocs$ such that $\loc$ is reachable in~$\valuate{\A}{\pval}$.

A finite run is \emph{accepting} if its last state $(\loc, \clockval)$ is such that $\loc \in \LocFinal$.
The (timed) \emph{language} $\Lg(\valuate{\A}{\pval})$ is defined to be the set of timed words associated with all accepting runs of~$\valuate{\A}{\pval}$.
% $\{\word \mid \text{there is an accepting run of $\A$ over $\word$}\}$ of timed words.

%%%%%%%%%%%%%%%%%%%%%%%%%%%%%%%%%%%%%%%%%%%%%%%%%%%%%%%%%%%%
\subsection{Reachability synthesis}
%%%%%%%%%%%%%%%%%%%%%%%%%%%%%%%%%%%%%%%%%%%%%%%%%%%%%%%%%%%%

We will use here reachability synthesis to solve parametric timed pattern matching.
This procedure, called \EFsynth{}, takes as input a PTA~$\A$ and a set of target locations~$\somelocs$, and attempts to synthesize all parameter valuations~$\pval$ for which~$\somelocs$ is reachable in~$\valuate{\A}{\pval}$.
\EFsynth{} was formalized in \eg{} \cite{JLR15} and is a procedure that may not terminate, but that computes an exact result (sound and complete) if it terminates.
\EFsynth{} traverses the \emph{parametric zone graph} of~$\A$, which is a potentially infinite extension of the well-known zone graph of TAs (see, \eg{} \cite{ACEF09,JLR15}\LongVersion{ for a formal definition}).

%%%%%%%%%%%%%%%%%%%%%%%%%%%%%%%%%%%%%%%%%%%%%%%%%%%%%%%%%%%%
\subsection{Timed pattern matching}
%%%%%%%%%%%%%%%%%%%%%%%%%%%%%%%%%%%%%%%%%%%%%%%%%%%%%%%%%%%%

Let us recall timed pattern matching~\cite{WAH16,WHS17}\todo{more refs}.

\smallskip

\defProblem
	{Timed pattern matching}
	{a TA~$\A$, a timed word~$\str$ over a common alphabet $\Actions$}
	{compute all the intervals $(t,t')$ for which the segment  $\word|_{(t,t')}$ is 
	accepted by $\A$.
	That is, it requires
	the \emph{match set} $\mathcal{M}
	(\word,\A) = \{(t,t') \mid \word|_{(t,t')} \in \Lg(\A)\}$.}

\medskip

% \begin{definition}[]\label{def:TimedPatternMatching}
% 	Let $\A$
% 	be a timed automaton, and  $\str$ be a timed word,  over a common  alphabet $\Actions$. The \emph{timed pattern matching} problem requires
% 	
% \end{definition}

The match set $\mathcal{M} (\word,\A)$ is in general uncountable; however it allows finite representation, as a finite union of special polyhedra called \emph{zones} (see~\cite{BY03,WAH16}).

We now extend this problem to parameters by allowing a specification expressed using PTAs.
The problem now requires not only the start and end dates for which the property holds, but also the associated parameter valuations.

\smallskip

\defProblem
	{Parametric timed pattern matching}
	{a PTA~$\A$, a timed word~$\str$ over a common alphabet $\Actions$}
	{compute all the triples $(t,t', \pval)$ for which the segment  $\word|_{(t,t')}$ is accepted by $\valuate{\A}{\pval}$.
	That is, it requires the \emph{match set} $\mathcal{M} (\word,\A) = \{(t,t', \pval) \mid \word|_{(t,t')} \in \Lg(\valuate{\A}{\pval})\}$.}

\medskip

We will see that this match set can still be represented as a finite union of polyhedra, but in more dimensions, \viz{} $|\Param| + 2$, \ie{} the number of parameters + 2 further dimensions for~$t$ and~$t'$.
However, the form of the obtained polyhedra is more general than zones, as parameters may ``accumulate'' to produce sums of parameters with coefficients (\eg{} $3 \times \param_1 < \param_2 + 2 \times \param_3$).

%%%%%%%%%%%%%%%%%%%%%%%%%%%%%%%%%%%%%%%%%%%%%%%%%%%%%%%%%%%%
%%%%%%%%%%%%%%%%%%%%%%%%%%%%%%%%%%%%%%%%%%%%%%%%%%%%%%%%%%%%
\section{Timed pattern matching under uncertainty}\label{section:approach}
%%%%%%%%%%%%%%%%%%%%%%%%%%%%%%%%%%%%%%%%%%%%%%%%%%%%%%%%%%%%
%%%%%%%%%%%%%%%%%%%%%%%%%%%%%%%%%%%%%%%%%%%%%%%%%%%%%%%%%%%%

%%%%%%%%%%%%%%%%%%%%%%%%%%%%%%%%%%%%%%%%%%%%%%%%%%%%%%%%%%%%
\subsection{General approach}\label{ss:general}
%%%%%%%%%%%%%%%%%%%%%%%%%%%%%%%%%%%%%%%%%%%%%%%%%%%%%%%%%%%%

We make the following two assumptions, that do not impact the correctness of our method, but simplify the subsequent reasoning.

\begin{assumption}\label{assumption:final}
	As in~\cite{WAH16,WHS17}, we assume that the pattern automaton contains a single final location, and that all transitions to this final location are labeled with~$\$$.
\end{assumption}

\begin{assumption}\label{assumption:alphabet}
	We assume that the pattern automaton contains no action not part of the timed word alphabet, with the exception of the special action~$\$$.
\end{assumption}

Both assumptions are easy to remove in practice: for \cref{assumption:final}, if the pattern PTA contains more than one final location, they can be merged into a single final location.
For \cref{assumption:alphabet}, the transitions labeled with actions not part of the timed word alphabet can simply be deleted.

We show using the following approach that parametric timed pattern matching can reduce to parametric reachability analysis.

\begin{enumerate}
	\item We turn the pattern into a symbolic pattern, by allowing it to start anytime.
		In addition, we use two parameters to measure the (symbolic) starting time and the (symbolic) ending time of the pattern.
	\item We turn the timed word into a (non-parametric) timed automaton, that uses a single clock $\clockabs$, that measures the absolute time.
	\item We consider the synchronized product of the symbolic pattern PTA and the timed word (P)TA.
	\item We run the reachability synthesis algorithm to derive all possible parameter valuations for which the final location of the pattern automaton is reachable.
\end{enumerate}

%%%%%%%%%%%%%%%%%%%%%%%%%%%%%%%%%%%%%%%%%%%%%%%%%%%%%%%%%%%%
\subsection{Our approach step by step}
%%%%%%%%%%%%%%%%%%%%%%%%%%%%%%%%%%%%%%%%%%%%%%%%%%%%%%%%%%%%

%%%%%%%%%%%%%%%%%%%%%%%%%%%%%%%%%%%%%%%%%%%%%%%%%%%%%%%%%%%%
\begin{figure*}[tb]
	\begin{subfigure}[b]{\textwidth}
	\centering
		\scalebox{1}{
		\begin{tikzpicture}[shorten >=1pt,node distance=2cm,on grid,auto]
		%% states
		\node[location,initial] (prepres_0) {$\loc_0''$};
		\node[location, above of=prepres_0] (pres_0) {$\loc_0'$};
		\node[location, below right of=prepres_0] (s_0) {$\loc_0$};
		\node[location] (s_1) [right of=s_0] {$\loc_1$};
		\node[location] (s_2) [right of=s_1] {$\loc_2$};
		\node[location] (s_3) [right of=s_2]{$\loc_3$};
		\node[location] (s_4) [right of=s_3]{$\loc_4$};
		\node[location,accepting] (s_5) [right of=s_4]{$\loc_5$};

		%% edges
		\path[->] 
		(prepres_0) edge node[below left,align=center] {$\styleclock{\clockabs} = \styleparam{t} = 0$ \\ \styleact{start}} (s_0)
		(prepres_0) edge[] node[align=center] {$\styleact{a},\styleact{b}$ \\ $\styleclock{x} := 0$} (pres_0)
		(pres_0) edge [loop left] node[left,align=center] {$\styleact{a},\styleact{b}$ \\ $\styleclock{x} := 0$} (pres_0)
		(pres_0) edge [bend left,above right] node[align=center] {$\styleclock{\clockabs} = \styleparam{t} \land \styleclock{x} > 0$ \\ \styleact{start} \\ $\styleclock{x} := 0$} (s_0)
		(s_0) edge [above] node {\begin{tabular}{c}
									$\styleclock{x} > \styleparam{\param_1}$\\
									$\styleact{a}$\\
									$\styleclock{x} := 0$
								\end{tabular}} (s_1)
		(s_1) edge [above] node {\begin{tabular}{c}
									$\styleclock{x} < \styleparam{\styleparam{\param_2}}$\\
									$\styleact{a}$\\
									$\styleclock{x} := 0$
								\end{tabular}} (s_2)
		(s_2) edge [above] node[align=center] {$\styleclock{x} < \styleparam{\styleparam{\param_2}}$\\$\styleact{a}$} (s_3)
		(s_3) edge [above] node[align=center] {$\styleclock{\clockabs} = \styleparam{t'}$\\$\styleact{\$}$\\$\styleclock{x} := 0$} (s_4)
		(s_4) edge [above] node[align=center] {$\styleclock{x} > 0$} (s_5)
		;
		\end{tikzpicture}}
	\caption{$\TransPattern$ applied to the PTA in \cref{figure:example:PTA}}
	\label{figure:approach:PTA}
	\end{subfigure}
	%------------------------------------------------------------
	
	%------------------------------------------------------------
	\begin{subfigure}[b]{0.99\textwidth}
	\centering
	\footnotesize
		\begin{tikzpicture}[shorten >=1pt,node distance=1.9cm,on grid,auto]
		%% states
		\node[location,initial] (w0) {$\wloc_0$};
		\node[location] (w1) [right of=w0] {$\wloc_1$};
		\node[location] (w2) [right of=w1] {$\wloc_2$};
		\node[location] (w3) [right of=w2]{$\wloc_3$};
		\node[location] (w4) [right of=w3]{$\wloc_4$};
		\node[location] (w5) [right of=w4]{$\wloc_5$};
		\node[location] (w6) [right of=w5]{$\wloc_6$};
		\node[location] (w7) [right of=w6]{$\wloc_7$};
		\node[location] (w8) [right of=w7]{$\wloc_8$};
		\node[location] (w9) [right of=w8]{$\wloc_9$};

		%% edges
		\path[->]
			(w0) edge [above] node[align=center] {$\styleclock{\clockabs} = 0.5$\\$\styleact{a}$} (w1)
			(w1) edge [above] node[align=center] {$\styleclock{\clockabs} = 0.9$\\$\styleact{a}$} (w2)
			(w2) edge [above] node[align=center] {$\styleclock{\clockabs} = 1.3$\\$\styleact{b}$} (w3)
			(w3) edge [above] node[align=center] {$\styleclock{\clockabs} = 1.7$\\$\styleact{b}$} (w4)
			(w4) edge [above] node[align=center] {$\styleclock{\clockabs} = 2.8$\\$\styleact{a}$} (w5)
			(w5) edge [above] node[align=center] {$\styleclock{\clockabs} = 3.7$\\$\styleact{a}$} (w6)
			(w6) edge [above] node[align=center] {$\styleclock{\clockabs} = 4.9$\\$\styleact{a}$} (w7)
			(w7) edge [above] node[align=center] {$\styleclock{\clockabs} = 5.3$\\$\styleact{a}$} (w8)
			(w8) edge [above] node[align=center] {$\styleclock{\clockabs} = 6.0$\\$\styleact{a}$} (w9)
		;
		\end{tikzpicture}
	\caption{$\TransWord$ applied to the timed word in \cref{figure:example:word}}
	\label{figure:approach:word}
	\end{subfigure}

	\caption{Our transformations exemplified on \cref{figure:example}}
	\label{figure:approach}
\end{figure*}
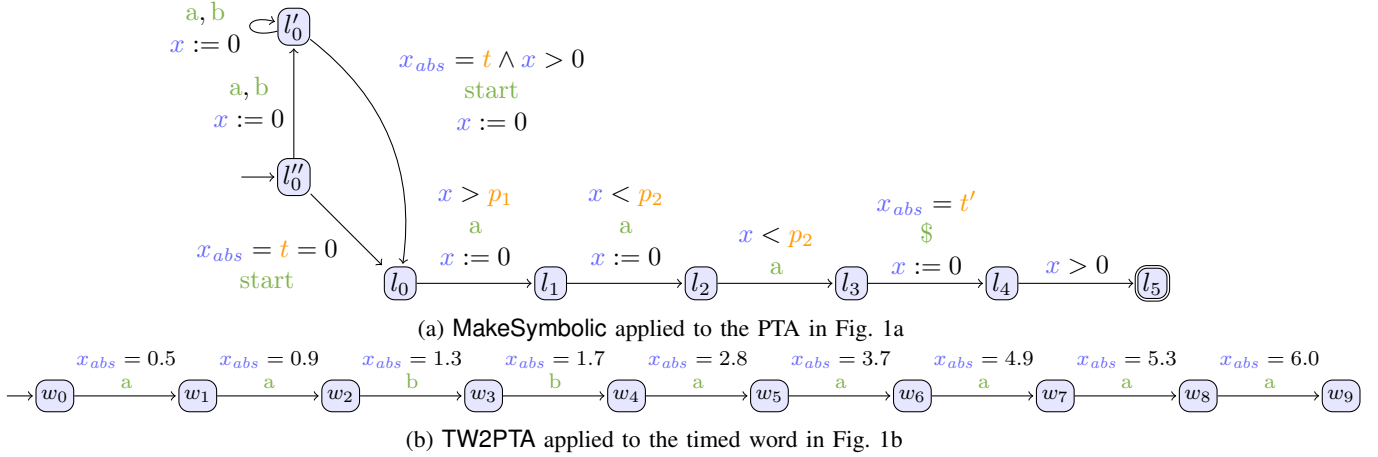
%%%%%%%%%%%%%%%%%%%%%%%%%%%%%%%%%%%%%%%%%%%%%%%%%%%%%%%%%%%%

\subsubsection{Making the pattern symbolic}
In this first step, we first add two parameters $t$ and~$t'$, which encode the (symbolic) start and end time where the pattern holds on the input timed word.
This way, we will obtain a result in the form of a finite union of polyhedra in $\ParamCard + 2$ dimensions, where the 2 additional dimensions come from the addition of~$t$ and~$t'$.
We also add a clock~$\clockabs$ initially~0 and never reset (this clock is shared by the pattern PTA and the subsequent timed word TA).
Then, we modify the pattern PTA as follows:
\begin{enumerate}
	\item we add two fresh locations (say~$\locinit'$ and $\locinit''$) prior to the initial location~$\locinit$;
	\item we add a fresh clock (say~$\clock$); in practice, as this clock is used only in the initial location, an existing clock of the pattern may be reused;
	\item we add an unguarded self-loop allowing any action of the timed word on $\locinit'$, and resetting~$\clock$;
	\item we add an unguarded transition from~$\locinit''$ to $\locinit'$ allowing any action of the timed word on $\locinit'$, and resetting~$\clock$;
	\item we add a transition from~$\locinit'$ to~$\locinit$ guarded by~$\clockabs = t \land \clock > 0$, labeled with a fresh action \styleact{start} and resetting all clocks of the pattern (except~$\clockabs$);
	\item we add a transition from~$\locinit''$ to~$\locinit$ guarded by~$\clockabs = t \land \clockabs = 0$, labeled with \styleact{start};
	\item we add a guard $\clockabs = t'$ and reset $\clock$ on the final transitions labeled with~$\$$;
	\item we add an extra location after the former final location, with a transition guarded by $\clock > 0$;
	\item the initial location of the modified PTA becomes~$\locinit''$.
% 		\ea{I realize the final transition should \emph{also} be followed by an additional transition either unguarded (if last action of the word) or guarded with any action of the word. FIXED}
\end{enumerate}

Let us give the intuition behind our transformation.
First, the two guards $\clockabs = t$ and $\clockabs = t'$ allow to record symbolically the value of the starting and ending dates.
Second, the new locations $\locinit''$ and $\locinit'$ allow the pattern to ``start anytime''; that is, it can synchronize with the timed word TA for an arbitrary long time while staying in the initial location~$\locinit''$ (and therefore \emph{not} matching the pattern), and start (using the transition from~$\locinit'$ to~$\locinit$) anytime.
% The constraint $\clock > 0$ ensures that the pattern starts a non-zero time after the previous action (if any).
Third, due to the constraint $\clock > 0$, a non-zero time must elapse between the last action before the pattern start and the actual pattern start.
The distinction between $\locinit''$ and $\locinit'$ is necessary to also allow starting the pattern at $\clockabs = 0$ if no action occurred before.
% \footnote{%
% 	This constructions also forbids the pattern to start at $t = 0$, which is unwanted.
% 		In fact, it suffices to add a second transition from $\locinit'$ to~$\locinit$ that can be taken at $t = 0$ (hence without the constraint $\clock = 0$) only if no self-loop was taken first. This can be easily achieved by duplicating .... ?????
% }
Finally, the guard $\clock > 0$ just before the final location ensures the next action of the system (if any) is taken after a non-zero delay, following our definitions of timed word and projection.

Let \TransPattern{} denote this procedure.\LongVersion{

}
Consider the pattern PTA~$\A$ in \cref{figure:example:PTA}.
The result of $\TransPattern(\A)$ is given in \cref{figure:approach:PTA}.
Note that we use the same clock $\clock$ for both the extra clock introduced by our construction and the original clock of the pattern automaton from \cref{figure:example:PTA}.
\todo{gain vertical space here!}

\subsubsection{Converting the timed word into a (P)TA}
In this second step, we convert the timed word into a (non-parametric) timed automaton.
This is very straightforward, and simply consists in converting a timed word of the form $(\action_1, \tau_1), \cdots , (\action_n, \tau_n)$ into a sequence of transitions labeled with~$\action_i$ and guarded with $\clockabs = \tau_i$ (recall that $\clockabs$ measures the absolute time and is shared by the timed word automaton and the pattern automaton).

Let \TransWord{} denote this procedure.\LongVersion{

}
Consider the timed word~$\word$ in \cref{figure:example:word}.
The result of $\TransWord(\word)$ is given in \cref{figure:approach:word}.

\subsubsection{Synchronized product}
The last part of the method consists in performing the synchronized product of $\TransPattern(\A)$ and $\TransWord(\word)$, and calling \EFsynth{} on the resulting PTA.

\LongVersion{
	\medskip
}

We summarize our method $\PTPM(\A,\word)$ in \cref{algo:PTPM}.

%------------------------------------------------------------
\begin{algorithm}[tb]
	\Input{A pattern PTA $\A$ with final location~$\LocFinal$, a timed word~$\word$}
	\Output{Constraint $\K$ over the parameters}

	\BlankLine

	$\A' \assign \TransPattern(\A)$
	
	$\A_{\word} \assign \TransWord(\word)$
	
	\Return $\EFsynth(\A' \parallel \A_{\word}, \LocFinal)$
	
	\caption{$\PTPM(\A, \word)$}
	\label{algo:PTPM}
\end{algorithm}
%------------------------------------------------------------

\begin{example}
	Consider again the timed word~$\word$ and the PTA pattern~$\A$ in \cref{figure:example}.
	The result of $\PTPM(\A,\word)$ is as follows:
	
	%  10*tprime >= 49
	% & p1 >= 0
	% & 53 >= 10*tprime
	% & 10*t >= 17
	% & 14 > 5*p1 + 5*t
	% & 5*p2 > 6
	% OR
	%   6 >= tprime
	% & p1 >= 0
	% & 10*tprime >= 53
	% & 37 > 10*p1 + 10*t
	% & 5*t >= 14
	% & 5*p2 > 6
	% OR
	%   p1 >= 0
	% & 10*t >= 37
	% & 49 > 10*p1 + 10*t
	% & 10*p2 > 7
	% & tprime >= 6
	
	% NOTE: NEW VERSION JUNE 2018
% 	10*t > 17
% 	& 10*tprime >= 49
% 	& 53 > 10*tprime
% 	& p1 >= 0
% 	& 14 > 5*t + 5*p1
% 	& 5*p2 > 6
% 	OR
% 	6 > tprime
% 	& 10*tprime >= 53
% 	& 5*t > 14
% 	& 37 > 10*t + 10*p1
% 	& p1 >= 0
% 	& 5*p2 > 6
% 	OR
% 	10*p2 > 7
% 	& tprime >= 6
% 	& 49 > 10*t + 10*p1
% 	& p1 >= 0
% 	& 10*t > 37

	\begin{align*}
		&
		1.7 < \styleparam{t} < 2.8 - \styleparam{\param_1}
		\land
		4.9 \leq \styleparam{t'} < 5.3
		\land
		\styleparam{\param_2} > 1.2
		\\
		\lor\ \ &
		2.8 < \styleparam{t} < 3.7 - \styleparam{\param_1}
		\land
		5.3 \leq \styleparam{t'} < 6
		\land
		\styleparam{\param_2} > 1.2
		\\
		\lor \ \ &
		3. 7 < \styleparam{t} < 4.9 - \styleparam{\param_1}
		\land
		\styleparam{t'} \geq 6
		\land
		\styleparam{\param_2} > 0.7
	\end{align*}
	Observe that, for the parameter valuation given in the introduction ($\param_1 = \param_2 = 1$), only the pattern corresponding to the last disjunct could be obtained, \ie{} the pattern that matches the last three $a$ of the timed word in \cref{figure:example:word}.
	In contrast, the first disjunct can match the first three $a$ coming after the two $b$s, while the second disjunct allows to match the three $a$s in the middle of the last five $a$s in \cref{figure:example:word}.\LongVersion{
	
	}
	We give various projections of this constraint onto two dimensions in \cref{figure:projections} (the difference between plain red and light red is not significant---light red constraints denote unbounded constraints towards at least one dimension).
\end{example}

%%%%%%%%%%%%%%%%%%%%%%%%%%%%%%%%%%%%%%%%%%%%%%%%%%%%%%%%%%%%
\begin{figure*}[t]
	\begin{subfigure}[b]{.3\textwidth}
		\centering
		
		\includegraphics[width=.8\textwidth]{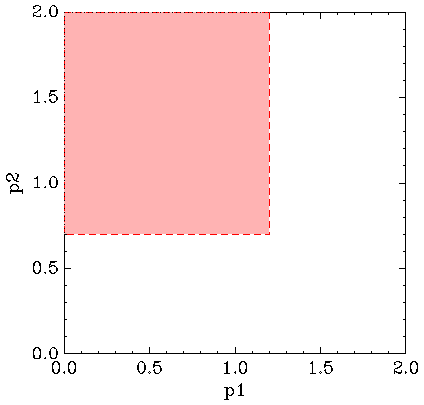}

		\caption{On $\param_1$ and $\param_2$}
	\end{subfigure}
	\hspace{1cm}
	\begin{subfigure}[b]{.3\textwidth}
		\centering
		
		\includegraphics[width=.8\textwidth]{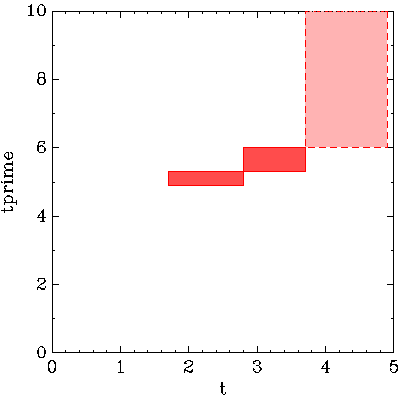}

		\caption{On $t$ and $t'$}
	\end{subfigure}
	\hspace{1cm}
	\begin{subfigure}[b]{.3\textwidth}
		\centering
		
		\includegraphics[width=.8\textwidth]{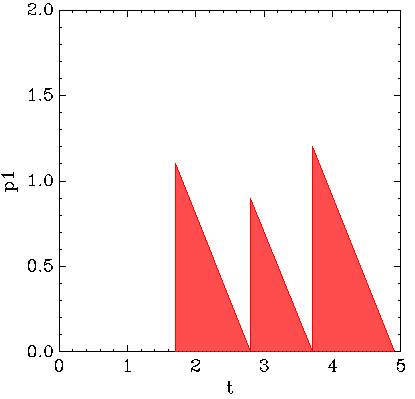}

		\caption{On $t$ and $\param_1$}
	\end{subfigure}
	\caption{Projections of the result of parametric timed pattern matching on \cref{figure:example}}
	\label{figure:projections}
\end{figure*}
%%%%%%%%%%%%%%%%%%%%%%%%%%%%%%%%%%%%%%%%%%%%%%%%%%%%%%%%%%%%

%%%%%%%%%%%%%%%%%%%%%%%%%%%%%%%%%%%%%%%%%%%%%%%%%%%%%%%%%%%%
\subsection{Termination}
%%%%%%%%%%%%%%%%%%%%%%%%%%%%%%%%%%%%%%%%%%%%%%%%%%%%%%%%%%%%

\LongVersion{
	We state below the termination of our procedure.
}

%------------------------------------------------------------
\begin{lemma}[termination]\label{lemma:termination}
	Let $\A$ be a PTA encoding a parametric pattern, and $\word$ be a timed word.
	Then $\PTPM(\A,\word)$ terminates.
\end{lemma}
%------------------------------------------------------------
\begin{proof}
	First, observe that there may be non-determinism in the pattern PTA, \ie{} the timed word can potentially synchronize with two transitions labeled with the same action from a given location.
	Even if there is no syntactic nondeterminism, nondeterminism can appear due to the interleaving of the initial \styleact{start} action:
		in \cref{figure:approach}, the first \styleact{a} of the timed word can either synchronize with the self-loop on~$\locinit'$, or the \styleact{start} action can first occur, and then the first \styleact{a} synchronizes of the timed word with the \styleact{a} labeling the transition from~$\locinit$ to~$\loc_1$ of the pattern PTA.
	Second, the pattern PTA may well have loops (and this is the case in our experiments in \cref{ss:XP:blowup}), which yields an infinite parametric zone graph (for the pattern automaton not synchronized with the word automaton).
	However, let us show that only a finite part of the parametric zone graph is explored by \EFsynth{}:
		indeed, since $\TransWord(\word)$ is only a finite sequence, and thanks to the strong synchronization between the pattern PTA and the timed word PTA and due to \cref{assumption:alphabet}, only a finite number of finite discrete paths in the synchronized product will be explored.
	The only interleaving is due to the initial \styleact{start} action (which appears twice in the pattern PTA but can only be taken once at most due to the mutually exclusive guards $\clock = 0$ and $\clock > 0$), and due to the final $\styleact{\$}$ action, that only appear on the last transition to the last-but-one final location.
	As the pattern PTA is finitely branching, this gives a finite number of finite paths.
	The length of each path is clearly bounded by~$|\word| + 3$.
	Let us now consider the maximal number of such paths: given a location in $\TransWord(\word)$, the choice of the action (say~$\styleact{\action}$) is entirely deterministic.
	However, the pattern PTA may be non-deterministic, and can synchronize with $\BranchingCard$ outgoing transitions labeled with~$\styleact{\action}$, which gives $\BranchingCard^{|\word|}$ combinations.
	In addition, the \styleact{start} action can be inserted exactly once, at any position in the timed word (from before the first action to after the last action of the word---in the case of an empty pattern): this gives therefore $(|\word|+1)\times\BranchingCard^{|\word|}$ different runs.
	(The $\styleact{\$}$ is necessarily the last-but-one action, and does not impact the number of runs, as the (potential) outgoing transitions from the final location are not explored.)
	Altogether, a total number of at most $(|\word| + 3) \times (|\word|+1) \times\BranchingCard^{|\word|}$ symbolic states is explored by \EFsynth{} in the worst case.
% 	\ea{ne pas oublier de dire qu'on supprime les transitions n'appartenant pas à l'intersection des deux automates (????)}
\end{proof}
%------------------------------------------------------------

\reviewer{2}{
* As far as I understand, the proof of Lemma 1 depends on the internal
of EFsynth and, without more description of EFsynth (e.g., how EFsynth
works or some sufficient condition on the input for EFsynth to
terminate), it would be hard to grasp it.}

\cref{lemma:termination} may not come as a surprise, as the input timed word is finite.
But it is worth noting that it comes in contrast with the fact that the wide majority of decision problems are undecidable for parametric timed automata, including the emptiness of the valuation set for which a given location is reachable both, for integer- and rational-valued parameters~\cite{AHV93,Miller00} (see \cite{Andre17STTT} for a survey).

\reviewer{2}{* First, there is no correctness argument for the proposed
algorithm (except for its termination), although it's intuitively
clear.  It would be nice to state correctness as a theorem.}

\todo{

We state below the correctness of our procedure.

%------------------------------------------------------------
\begin{lemma}[correctness]\label{lemma:correctness}
	\todo{}
\end{lemma}
%------------------------------------------------------------
\begin{proof}
	\todo{}
\end{proof}
%------------------------------------------------------------

}

%%%%%%%%%%%%%%%%%%%%%%%%%%%%%%%%%%%%%%%%%%%%%%%%%%%%%%%%%%%%
\subsection{Pattern matching with optimization}\label{ss:optimization}
%%%%%%%%%%%%%%%%%%%%%%%%%%%%%%%%%%%%%%%%%%%%%%%%%%%%%%%%%%%%

We also address the following optimization problem:
given a timed word and a pattern containing parameters, what is the minimum or maximum value of a given parameter such that the pattern is matched by the timed word?
That is, we are only interested in the optimal value of the given parameter, and not in the full list of matches as in~\PTPM{}.

While this problem can be solved using our solution from \cref{ss:general} (by computing the multidimensional constraint, and then eliminating all parameters but the target parameter, using variable elimination techniques), we use here a dedicated approach, with the hope it be more efficient.
Instead of managing all symbolic matches (\ie{} a finite union of polyhedra), we simply manage the current optimum; in addition, we cut branches that cannot improve the optimum, with the hope to reduce the number of states explored.
For example, assume parameter $\param$ is to be minimized; if the current minimum is $\param > 2$, and if a branch is such that $\param \geq 3$, then this branch will not improve the minimum, and can safely be discarded.
Let \PTPMopt{} denote this procedure.

\begin{example}
	Consider again the timed word~$\word$ and the PTA pattern~$\A$ in \cref{figure:example}.
	Minimizing $\styleparam{\param_2}$ so that the pattern matches the timed word for at least one position gives
	$\styleparam{\param_2} > 0.7$,
	while maximizing $\styleparam{\param_1}$ gives $\styleparam{\param_1} < 1.2$.
\end{example}

%%%%%%%%%%%%%%%%%%%%%%%%%%%%%%%%%%%%%%%%%%%%%%%%%%%%%%%%%%%%
%%%%%%%%%%%%%%%%%%%%%%%%%%%%%%%%%%%%%%%%%%%%%%%%%%%%%%%%%%%%
\section{Experiments}\label{section:experiments}
%%%%%%%%%%%%%%%%%%%%%%%%%%%%%%%%%%%%%%%%%%%%%%%%%%%%%%%%%%%%
%%%%%%%%%%%%%%%%%%%%%%%%%%%%%%%%%%%%%%%%%%%%%%%%%%%%%%%%%%%%

\reviewer{1}{It would have been also nice to have an experiment that would have shown how the computation time increases with the number of parameters to compute.}

We evaluated our approach against two standard benchmarks from~\cite{HAF14}, already used in~\cite{WHS17}, as well as a third benchmark specifically designed to test the limits of parametric timed pattern matching.
We fixed no bounds for our parameters.

We used \imitator{}~\cite{AFKS12} to perform the parameter synthesis (algorithm \EFsynth{}).
\imitator{} relies on the Parma Polyhedra Library (PPL)~\cite{BHZ08} to compute symbolic states.
It was shown in~\cite{BFMU17} that polyhedra may be dozens of times slower than more efficient data structures such as DBMs (difference bound matrices); however, for parametric analyses, DBMs are not suitable, and parameterized extensions (\eg{} in~\cite{HRSV02}) still need polyhedra in their representation.

\AuthorVersion{%
	We used a slightly modified version of \imitator{} for technical reasons (see \cref{appendix:polyhedra}).
}
We wrote a simple Python script to implement the \TransWord{} procedure; the patterns (\cref{figure:patterns}) were manually transformed following the \TransPattern{} procedure, and converted into the input language of \imitator{}.

We ran experiments using \imitator{} 2.10.4 ``Butter Jellyfish'' %(build 2475 \texttt{develop/861dbac})
on a Dell Precision 3620 i7-7700 3.60\,GHz with 64\,GiB memory running Linux Mint 19 beta 64\,bits.
Sources, binaries, models, logs can be found at \IEEEVersion{\texttt{www.imitator.fr/static/ICECCS18}}\AuthorVersion{\url{www.imitator.fr/static/ICECCS18}}.

%%%%%%%%%%%%%%%%%%%%%%%%%%%%%%%%%%%%%%%%%%%%%%%%%%%%%%%%%%%%
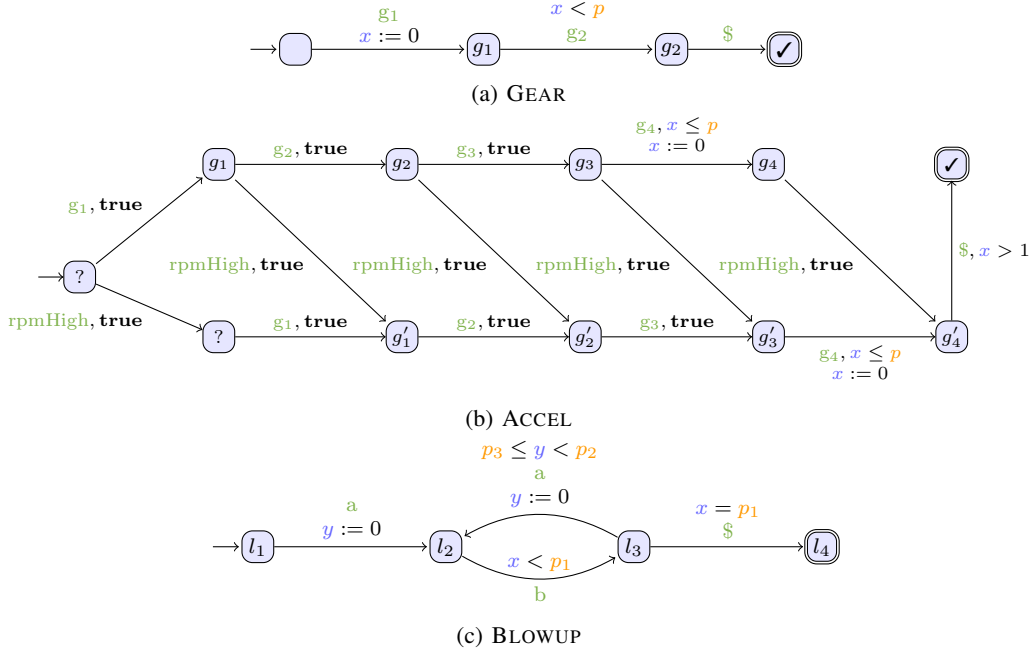
\begin{figure*}[t]
	\begin{subfigure}[b]{\textwidth}
	\centering
	\footnotesize

  \begin{tikzpicture}[shorten >=1pt,node distance=2.5cm,on grid,auto] 
    \node[location,initial] (s_0)  {}; 
    \node[location,node distance=2.5cm] (s_1) [right=of s_0] {$g_1$}; 
    \node[location,node distance=2.5cm] (s_2) [right=of s_1] {$g_2$};
    \node[location,accepting,node distance=1.5cm] (s_3) [right=of s_2] {\cmark};
    \path[->] 
    (s_0) edge  [above] node[align=center] {$\styleact{g_1}$\\$\styleclock{x} := 0$} (s_1)
    (s_1) edge  [above] node[align=center] {$\styleclock{x} < \styleparam{\param} $\\$\styleact{g_2}$} (s_2) % NOTE: originally 2
    (s_2) edge  [above] node {$\styleact{\$}$} (s_3);
  \end{tikzpicture}
  \caption{\textsc{Gear}}
%   \caption{\textsc{Gear}. The set $W$ (length 307--1,011,427) is generated by the automatic transmission system model in~\cite{DBLP:conf/cpsweek/HoxhaAF14}. The pattern, from $\phi^{\mathit{AT}}_5$ in~\cite{DBLP:conf/cpsweek/HoxhaAF14}, is for an event in which gear shift occurs too quickly (from the 1st to 2nd).
% }
	\label{figure:patterns:gear}
	\end{subfigure}
	%------------------------------------------------------------
	
	%------------------------------------------------------------
	\begin{subfigure}[b]{\textwidth}
	\centering
	\scriptsize

	\begin{tikzpicture}[auto,node distance=2cm]
		\node[location, initial] (s_000) {$?$};

		\node[location] (s_100)[above right=of s_000,yshift=-10] {$g_1$};
		\node[location] (s_001)[below right=of s_000,yshift=30] {$?$};

		\node[location] (s_200)[right=of s_100] {$g_2$};
		\node[location] (s_101)[right=of s_001] {$g_1'$};

		\node[location] (s_300)[right=of s_200] {$g_3$};
		\node[location] (s_201)[right=of s_101] {$g_2'$};

		\node[location] (s_400)[right=of s_300] {$g_4$};
		\node[location] (s_301)[right=of s_201] {$g_3'$};

		\node[location] (s_401)[right=of s_301] {$g_4'$};

		\node[location,accepting] (f)[right=of s_400] {\cmark};

		\path[->]
		(s_000) edge  [above left] node {$\styleact{g_1}, \CTrue$} (s_100)
		(s_100) edge  [above] node {$\styleact{g_2}, \CTrue$} (s_200)
		(s_200) edge  [above] node {$\styleact{g_3}, \CTrue$} (s_300)
		(s_300) edge  [above] node {\begin{tabular}{c}
										$\styleact{g_4}, \styleclock{x} \leq \styleparam{\param}$\\ % NOTE: originally 10
										$\styleclock{x} := 0$
									\end{tabular}} (s_400)

		(s_100) edge  [below left] node {$\styleact{rpmHigh}, \CTrue$} (s_101)
		(s_200) edge  [below left] node {$\styleact{rpmHigh}, \CTrue$} (s_201)
		(s_300) edge  [below left] node {$\styleact{rpmHigh}, \CTrue$} (s_301)
		(s_400) edge  [below left] node {$\styleact{rpmHigh}, \CTrue$} (s_401)

		(s_001) edge  [above] node {$\styleact{g_1}, \CTrue$} (s_101)
		(s_101) edge  [above] node {$\styleact{g_2}, \CTrue$} (s_201)
		(s_201) edge  [above] node {$\styleact{g_3}, \CTrue$} (s_301)
		(s_301) edge  [below] node {\begin{tabular}{c}
										$\styleact{g_4}, \styleclock{x} \leq \styleparam{\param}$\\ % NOTE: originally 10
										$\styleclock{x} := 0$
									\end{tabular}} (s_401)

		(s_000) edge  [below left] node {$\styleact{rpmHigh}, \CTrue$} (s_001)

		(s_401) [bend right=0] edge [right] node {$\styleact{\$}, \styleclock{x} > 1$} (f);
	\end{tikzpicture}
%   \caption{\textsc{Accel}.
% The set $W$ (length 25,002--17,280,002) is generated by the same automatic transmission system model as in \textsc{Gear}. The pattern
% is from
%  $\phi^{\mathit{AT}}_8$ 
% in~\cite{DBLP:conf/cpsweek/HoxhaAF14}: although the gear shifts from 1st to 4th and RPM is high enough somewhere in its course, the vehicle velocity is not high enough (i.e.\ the character veloHigh is absent). 
%  }

	\caption{\textsc{Accel}}
	\label{figure:patterns:accel}
	\end{subfigure}
	%------------------------------------------------------------

	%------------------------------------------------------------
	\begin{subfigure}[b]{\textwidth}
	\centering
	\footnotesize

	\begin{tikzpicture}[shorten >=1pt,node distance=2.5cm,on grid,auto] 
		\node[location,initial] (s_0)  {$\loc_1$}; 
		\node[location] (s_1) [right=of s_0] {$\loc_2$}; 
		\node[location] (s_2) [right=of s_1] {$\loc_3$};
		\node[location,accepting] (s_3) [right=of s_2] {$\loc_4$};
		\path[->] 
			(s_0) edge [above] node[align=center] {$\styleact{a}$\\$\styleclock{y} := 0$} (s_1)
			(s_1) edge[bend right] node[above] {$\styleclock{x} < \styleparam{\param_1}$} node[below] {$\styleact{b}$} (s_2)
			(s_2) edge node[above,align=center] {$\styleclock{x} = \styleparam{\param_1}$\\$\styleact{\$}$} (s_3)
			(s_2) edge[bend right] node[above,align=center] {$\styleparam{\param_3} \leq \styleclock{y} < \styleparam{\styleparam{\param_2}}$\\$\styleact{a}$\\$\styleclock{y} := 0$} (s_1)
		% 		(s_2) edge  [above] node {$\styleact{\$}$} (s_3)
		;
	\end{tikzpicture}
  \caption{\textsc{Blowup}}
	\label{figure:patterns:blowup}
	\end{subfigure}
	%------------------------------------------------------------

	\caption{Experiments: patterns}
	\label{figure:patterns}
\end{figure*}
%%%%%%%%%%%%%%%%%%%%%%%%%%%%%%%%%%%%%%%%%%%%%%%%%%%%%%%%%%%%

% \todo{number of parameters (1 in both case studies??) In fact useless, as we give the patterns in \cref{figure:patterns}}

\subsection{\textsc{Gear}}\label{ss:XP:gear}

Benchmark \textsc{Gear} is \LongVersion{inspired by the scenario of }monitoring the gear change of an automatic transmission system.
% We constructed the set $W$ as follows.
We conducted simulation of the model of an automatic transmission system~\cite{HAF14}.
We used S-TaLiRo~\cite{ALFS11} to generate an input to this model; it generates a gear change signal that is fed to the model.
A gear is chosen from $\{g_1,g_2,g_3,g_4\}$.
The generated gear change is recorded in a timed word.
\LongVersion{%
	The set $W$ consists of 10 timed words;
	% 	the length of each word is 307 to 1,011,427.
	the length of each word is 1,467 to 14,657.
}

The pattern PTA $\A$, shown in \cref{figure:patterns:gear}, detects the violation of the following condition:
If the gear is changed to~1, it should not be changed to 2 within $\param$ seconds.
This condition is related to the requirement $\phi^{\mathit{AT}}_5$ proposed in~\cite{HAF14} (the nominal value for~$\param$ in~\cite{HAF14} is~2).

We tabulate our experiments in \cref{table:gear}.
We give from left to right the length of the timed word in terms of actions and time, then the data for \PTPM{} (the number of symbolic states explored, the number of (symbolic) matches found, the parsing time and the computation time excluding parsing) and for \PTPMopt{} (number of symbolic states explored and computation time) using \imitator{}.
The parsing time for \PTPMopt{} is almost identical to~\PTPM{} and is therefore omitted.

%------------------------------------------------------------
\begin{table*}[tb]
	\centering
	\scriptsize
	
	\begin{tabular}{| r | r | r | r | r | r | r | r |}
		\hline
		\startMultiCellHeader{2}{Model} & \multiCellHeader{4}{\PTPM{}} & \multiCellHeader{2}{\PTPMopt}\\
% 		\hline
		\cellHeader{Length} & \cellHeader{Time frame} & \cellHeader{States} & \cellHeader{Matches} & \cellHeader{Parsing (s)} & \cellHeader{Comp.\ (s)} & \cellHeader{States} & \cellHeader{Comp.\ (s)}
		\\
		\hline
		1,467 & 1,000 & 4,453 & 379 & 0.02 & 1.60 & 3,322 & 0.94
		\\
		\hline
		2,837 & 2,000 & 8,633 & 739 & 0.33 & 2.14 & 6,422 & 1.70
		\\
		\hline
		4,595 & 3,000 & 14,181 & 1,247 & 0.77 & 3.63 & 10,448 & 2.85
		\\
		\hline
		5,839 & 4,000 & 17,865 & 1,546 & 1.23 & 4.68 & 13,233 & 3.74
		\\
		\hline
		7,301 & 5,000 & 22,501 & 1,974 & 1.94 & 5.88 & 16,585 & 4.79
		\\
		\hline
		8,995 & 6,000 & 27,609 & 2,404 & 2.96 & 7.28 & 20,413 & 5.76
		\\
		\hline
		10,316 & 7,000 & 31,753 & 2,780 & 4.00 & 8.38 & 23,419 & 6.86
		\\
		\hline
		11,831 & 8,000 & 36,301 & 3,159 & 5.39 & 9.75 & 26,832 & 7.87
		\\
		\hline
		13,183 & 9,000 & 40,025 & 3,414 & 6.86 & 10.89 & 29,791 & 8.61
		\\
		\hline
		14,657 & 10,000 & 44,581 & 3,816 & 8.70 & 12.15 & 33,141 & 9.89
		\\
		\hline
	\end{tabular}

	\caption{Experiments: \textsc{Gear}}
	\label{table:gear}
\end{table*}
%------------------------------------------------------------

The corresponding chart is given in \cref{figure:experiments:charts:gear} (\PTPM{} is given in plain black, and \PTPMopt{} in red dashed).
\PTPMopt{} brings a gain in terms of memory (symbolic states) of about 25\,\%, while the gain in time is about~20\,\%.

%%%%%%%%%%%%%%%%%%%%%%%%%%%%%%%%%%%%%%%%%%%%%%%%%%%%%%%%%%%%
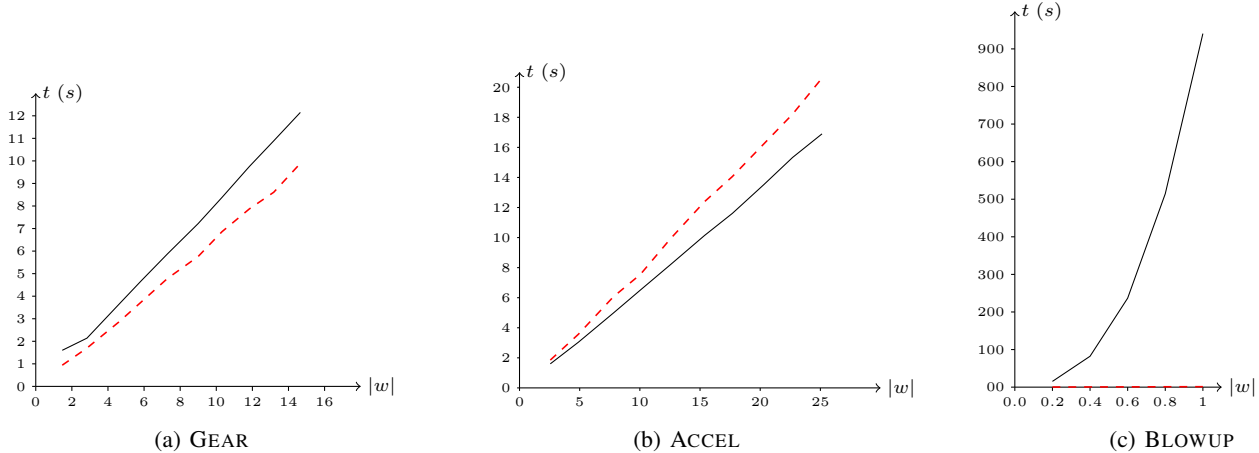
\begin{figure*}[t]
	\scriptsize

	%------------------------------------------------------------
	\begin{subfigure}[b]{.3\textwidth}
% 	\centering

	\begin{tikzpicture}[scale=0.3, xscale=.8]
		\draw[->] (0.0, 0.0) --++ (18.0, 0.0) node[right]{$|\word|$};
		\draw[->] (0.0, 0.0) --++ (0.0, 13.0) node[right]{$t$ $(s)$};
		
		% AXES
		\foreach \x in {0, 2, ..., 16} % X
			\draw (\x, 0) -- (\x, -.2)node [below] {\tiny{$\x$}};
		\foreach \x in {0, 1, ..., 12} % Y
			\draw (0, \x) -- (-.2, \x) node [left] {\tiny{$\x$}};

		% PTPM regular
		\draw[-]
			(1.467, 1.60)
			--
			(2.837, 2.14)
			--
			(4.595, 3.63)
			--
			(5.839, 4.68)
			--
			(7.301, 5.88)
			--
			(8.995, 7.22)
			--
			(10.316, 8.38)
			--
			(11.831, 9.75)
			--
			(13.183, 10.89)
			--
			(14.657, 12.15)
		;
		
		% OPT
		\draw[PTPMOPT]
			(1.467, 0.94)
			--
			(2.837, 1.70)
			--
			(4.595, 2.85)
			--
			(5.839, 3.74)
			--
			(7.301, 4.79)
			--
			(8.995, 5.76)
			--
			(10.316, 6.86)
			--
			(11.831, 7.87)
			--
			(13.183, 8.61)
			--
			(14.657, 9.89)
		;
	\end{tikzpicture}
	
	\caption{\textsc{Gear}}
	\label{figure:experiments:charts:gear}
	\end{subfigure}
	%------------------------------------------------------------
% 	\hfill{}
	\hspace{1cm}
	%
	%------------------------------------------------------------
	\begin{subfigure}[b]{.3\textwidth}
% 	\centering

	\begin{tikzpicture}[scale=0.2, xscale=.8]
		\draw[->] (0.0, 0.0) --++ (30.0, 0.0) node[right]{$|\word|$};
		\draw[->] (0.0, 0.0) --++ (0.0, 21.0) node[right]{$t$ $(s)$};
		
		% AXES
		\foreach \x in {0, 5, ..., 25} % X
			\draw (\x, 0) -- (\x, -.2)node [below] {\tiny{$\x$}};
		\foreach \x in {0, 2, ..., 20} % Y
			\draw (0, \x) -- (-.2, \x) node [left] {\tiny{$\x$}};
		
		% PTPM regular
		\draw[-]
			(2.559, 1.60)
			--
			(4.894, 3.04)
			--
			(7.799, 4.98)
			--
			(10.045, 6.51)
			--
			(12.531, 8.19)
			--
			(15.375, 10.14)
			--
			(17.688, 11.61)
			--
			(20.299, 13.52)
			--
			(22.691, 15.33)
			--
			(25.137, 16.90)
		;
		
		% OPT
		\draw[PTPMOPT]
			(2.559, 1.85)
			--
			(4.894, 3.57)
			--
			(7.799, 6.06)
			--
			(10.045, 7.55)
			--
			(12.531, 9.91)
			--
			(15.375, 12.39)
			--
			(17.688, 14.06)
			--
			(20.299, 16.23)
			--
			(22.691, 18.21)
			--
			(25.137, 20.61)

		;
	\end{tikzpicture}
	
	\caption{\textsc{Accel}}
	\label{figure:experiments:charts:accel}
	\end{subfigure}
	%------------------------------------------------------------
	%
	\hspace{1cm}
	%
	%------------------------------------------------------------
	\begin{subfigure}[b]{.3\textwidth}
% 	\centering

	\begin{tikzpicture}[scale=0.05, xscale=5]
		\draw[->] (0.0, 0.0) --++ (11.0, 0.0) node[right]{$|\word|$};
		\draw[->] (0.0, 0.0) --++ (0.0, 100.0) node[right]{$t$ $(s)$};
		
		% AXES
		\foreach \x in {0, 2, ..., 8} % X
			\draw (\x, 0) -- (\x, -1)node [below] {\tiny{$0.\x$}};
		% HACK for 1000
		\draw (10, 0) -- (10, -1)node [below] {\tiny{$1$}};
		\foreach \x in {0, 10, ..., 90} % Y
			\draw (0, \x) -- (-.2, \x) node [left] {\tiny{$\x{}0$}};
		
		% PTPM regular
		\draw[-]
			(2, 1.53)
			--
			(4, 8.22)
			--
			(6, 23.68)
			--
			(8, 51.46)
			--
			(10, 94.07)
		;
		
		% OPT
		\draw[PTPMOPT]
			(2, 0.024)
			--
			(4, 0.049)
			--
			(6, 0.071)
			--
			(8, 0.105)
			--
			(10, 0.124)
		;
	\end{tikzpicture}
	
	\caption{\textsc{Blowup}}
	\label{figure:experiments:charts:blowup}
	\end{subfigure}
	%------------------------------------------------------------
	
	\caption{Experiments: charts ($x$-scale $\times 1,000$)}
	\label{figure:experiments:charts}
\end{figure*}
%%%%%%%%%%%%%%%%%%%%%%%%%%%%%%%%%%%%%%%%%%%%%%%%%%%%%%%%%%%%

\subsection{\textsc{Accel}}\label{ss:XP:accel}

The $W$ of benchmark $\textsc{Accel}$ is also constructed from the
Simulink model of the automated transmission
system~\cite{HAF14}.  For this benchmark, the
(discretized) value of three state variables are recorded in $W$:
engine RPM (discretized to ``high'' and ``low'' with a certain
threshold), velocity (discretized to ``high'' and ``low'' with a
certain threshold), and 4 gear positions.
We used S-TaLiRo~\cite{ALFS11} to generate a input sequence of gear change.
\LongVersion{%
Our set $W$ consists of 10 timed words;
% The length of each word is 25,002 to 17,280,002.
the length of each word is 2,559 to 25,137.
}

The pattern PTA $\A$ of this benchmark is shown in \cref{figure:patterns:accel}.
This pattern matches a part of a
timed word that violates the following condition: If a gear changes
from 1 to 2, 3, and~4 in this order in $\param$~seconds and engine RPM
becomes large during this gear change, then the velocity of the car
must be sufficiently large in one second.
This condition models the requirement $\phi^{\mathit{AT}}_8$ proposed in~\cite{HAF14} (the nominal value for~$\param$ in~\cite{HAF14} is~10).

%------------------------------------------------------------
\begin{table*}[tb]
	\centering
	\scriptsize
	
	\begin{tabular}{| r | r | r | r | r | r | r | r |}
		\hline
		\startMultiCellHeader{2}{Model} & \multiCellHeader{4}{\PTPM{}} & \multiCellHeader{2}{\PTPMopt}\\
% 		\hline
		\cellHeader{Length} & \cellHeader{Time frame} & \cellHeader{States} & \cellHeader{Matches} & \cellHeader{Parsing (s)} & \cellHeader{Comp.\ (s)} & \cellHeader{States} & \cellHeader{Comp.\ (s)}
		\\
		\hline
		2,559 & 1,000 & 6,504 & 2 & 0.27 & 1.60 & 6,502 & 1.85
		\\
		\hline
		4,894 & 2,000 & 12,429 & 2 & 0.86 & 3.04 & 12,426 & 3.57
		\\
		\hline
		7,799 & 3,000 & 19,922 & 7 & 2.21 & 4.98 & 19,908 & 6.06
		\\
		\hline
		10,045 & 4,000 & 25,520 & 3 & 3.74 & 6.51 & 25,514 & 7.55
		\\
		\hline
		12,531 & 5,000 & 31,951 & 9 & 6.01 & 8.19 & 31,926 & 9.91
		\\
		\hline
		15,375 & 6,000 & 39,152 & 7 & 9.68 & 10.14 & 39,129 & 12.39
		\\
		\hline
		17,688 & 7,000 & 45,065 & 9 & 13.40 & 11.61 & 45,039 & 14.06
		\\
		\hline
		20,299 & 8,000 & 51,660 & 10 & 18.45 & 13.52 & 51,629 & 16.23
		\\
		\hline
		22,691 & 9,000 & 57,534 & 11 & 24.33 & 15.33 & 57,506 & 18.21
		\\
		\hline
		25,137 & 10,000 & 63,773 & 13 & 31.35 & 16.90 & 63,739 & 20.61 % VERSION FORMATS: 1201.2 / 1202.4
		\\
		\hline
	\end{tabular}

	\caption{Experiments: \textsc{Accel}}
	\label{table:accel}
\end{table*}
%------------------------------------------------------------

Experiments are tabulated in \cref{table:accel}.
The corresponding chart is given in \cref{figure:experiments:charts:accel}.
This time, \PTPMopt{} brings almost no gain in terms of states, and a loss of speed of about 15 to~20\,\%, which may come from the additional polyhedra inclusion checks to test whether a branch is less good than the current optimum.

\subsection{\textsc{Blowup}}\label{ss:XP:blowup}

\ea{= looping case study studied in Tokyo with Masaki}

As a third experiment, we considered an original (toy) benchmark that acts as a worst case situation for parametric timed pattern matching.
Consider the PTA pattern in \cref{figure:patterns:blowup}, and assume a timed word consisting in an alternating sequence of ``$\styleact{a}$'' and ``$\styleact{b}$''.
Observe that the time from the pattern beginning (that resets~$\styleclock{x}$) to the end is exactly $\styleparam{\param_1}$ time units.
Also observe that the duration of the loop through $\loc_2$ and~$\loc_3$ has a duration in $[\styleparam{\param_3} , \styleparam{\param_2})$; therefore, for values sufficiently small of $\styleparam{\param_2},\styleparam{\param_3}$, one can always match a larger number of loops.
That is, for a timed word of length $2 n$ alternating between ``$\styleact{a}$'' and ``$\styleact{b}$'', there will be $n$ possible matches from position~0 (with $n$ different parameter constraints), $n-1$ from position 1, and so on, giving a total number of $\frac{n(n+1)}{2}$ matches with different constraints in 5~dimensions.

Note that this worst case situation is not specific to our approach, but would appear independently of the approach chosen for parametric timed pattern matching.
\LongVersion{

}%
We generated random timed words of various sizes, all alternating exactly between ``$\styleact{a}$'' and ``$\styleact{b}$''.
\LongVersion{Our set $W$ consists of 5 timed words of length from 200 to~1,000.}

% We also used a toy case study to study the performance in the presence of a blow-up model:

\todo{didn't include the final pairwise reduce in timings}

%------------------------------------------------------------
\begin{table*}[tb]
	\centering
	\scriptsize
	
	\begin{tabular}{| r | r | r | r | r | r | r | r |}
		\hline
		\startMultiCellHeader{2}{Model} & \multiCellHeader{4}{\PTPM{}} & \multiCellHeader{2}{\PTPMopt}\\
% 		\hline
		\cellHeader{Length} & \cellHeader{Time frame} & \cellHeader{States} & \cellHeader{Matches} & \cellHeader{Parsing (s)} & \cellHeader{Comp.\ (s)} & \cellHeader{States} & \cellHeader{Comp.\ (s)}
		\\
		\hline
		200 & 101 & 20,602 & 5,050 & 0.01 & 15.31 & 515 & 0.24
		\\
		\hline
		400 & 202 & 81,202 & 20,100 & 0.02 & 82.19 & 1,015 & 0.49
		\\
		\hline
		600 & 301 & 181,802 & 45,150 & 0.03 & 236.80 & 1,515 & 0.71
		\\
		\hline
		800 & 405 & 322,402 & 80,200 & 0.05 & 514.57 & 2,015 & 1.05
		\\
		\hline
		1,000 & 503 & 503,002 & 125,250 & 0.06 & 940.74 & 2,515 & 1.24 % TODO: redo to get .res (but timings are correct)
		\\
		\hline
	\end{tabular}

	\caption{Experiments: \textsc{Blowup}}
	\label{table:blowup}
\end{table*}
%------------------------------------------------------------

% NOTE: grep OR blowup-600.res -wc for most models
% NOTE: ????? for 

Experiments are tabulated in \cref{table:blowup}.
The corresponding chart is given in \cref{figure:experiments:charts:blowup}.
\PTPM{} becomes clearly non-linear as expected.
% 	(OoM denotes out of memory).
This time, \PTPMopt{} brings a dramatic gain in both memory and time; even more interesting, \PTPMopt{} remains perfectly linear.

\subsection{Discussion}

A first positive outcome is that our method is effectively able to perform parametric pattern matching on words of length up to several dozens of thousands, and is able to output results in the form of several dozens of thousands of symbolic matches in several dimensions, in just a few seconds.
%
% BEGIN FORMER VERSION WITH BAD EXPERIMENTS
% A disappointing outcome is that we expected \PTPM{} to be linear in the size of the input word for \textsc{Gear} and \textsc{Accel}: but, while the number of states explored by \PTPM{} is indeed perfectly linear\LongVersion{ for patterns without loops}, as one would expect from our method, \cref{figure:experiments:charts} shows that the computation time is clearly not.
% A quick investigation showed that the loss of speed comes from the size of the input model (the parsing speed is itself not linear at all):\LongVersion{ the internal data structures used by \imitator{} may be the reason why.}
% \imitator{} has to parse and manage a huge automaton (up to 25,137 locations), which is completely new compared to the usual usage of the tool, that mostly solved in the past complex benchmarks with a potentially large number of automata, but few (\ie{} dozens or at most hundreds of) locations per automaton.
% Therefore, further optimizations are needed in that direction.
% END FORMER VERSION WITH BAD EXPERIMENTS
Another positive outcome is that \PTPM{} is perfectly linear in the size of the input word for \textsc{Gear} and \textsc{Accel}: this was expected as these examples are linear, in the sense that the number of states explored by \PTPM{} is linear as these patterns feature no loops.

Note that the parsing time is not linear, but it could be highly improved: due to the relatively small size of the models usually treated by \imitator{}, this part was never properly optimized, and it contains several quadratic syntax checking functions.

The performances do not completely allow yet for an \emph{online} usage in the current version of our algorithm and implementation (in~\cite{WHS17}, we pushed the \textsc{Accel} case study for timed words of length up to~17,280,002).
A possible direction is to perform \LongVersion{and on-the-fly computation of the parametric zone graph, more precisely to do }an on-the-fly parsing of the timed word automaton; this will allow \imitator{} to keep in memory a single location at a time (instead of up to 25,137 in our experiments).

Finally, although this is not our original motivation, we believe that, if we are only interested in \emph{robust} pattern matching, \ie{} non-parametric pattern matching but with an allowed deviation (``guard enlargement'') of the pattern automaton, then using the efficient 1-dimensional parameterized DBMs of~\cite{Sankur15} would probably be an interesting alternative: indeed, in contrast to classical parameterized DBMs~\cite{HRSV02} (that are made of a matrix and a parametric polyhedron), the structure of~\cite{Sankur15} only needs an $\ClockCard\times\ClockCard$ matrix with a single parameter, and seems particularly efficient.

%%%%%%%%%%%%%%%%%%%%%%%%%%%%%%%%%%%%%%%%%%%%%%%%%%%%%%%%%%%%
\begin{figure*}[t]
	\begin{subfigure}[b]{.49\textwidth}
		\centering
		
		\includegraphics[width=.55\textwidth]{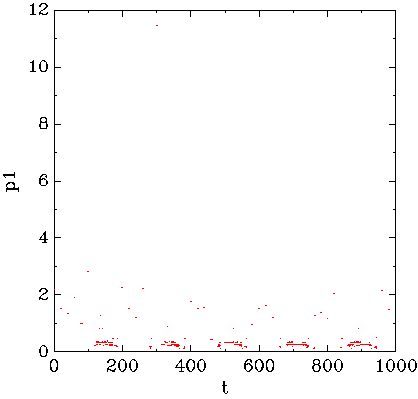}
		% NOTE: generated using ./imitator gear_1000.imi -mode EFunsafe -no-inclusion-test-in-EF -incl -verbose experiments -output-result -output-cart -output-cart-y-max 12 -output-cart-x-max 1000 (by setting 't, p1, tprime' as order for the declared parameter in the model)

		\caption{Projection onto $t$ and~$\param$}
		\label{figure:projections:gear:t-param}
	\end{subfigure}
	\begin{subfigure}[b]{.49\textwidth}
		\centering
		
		\includegraphics[width=.55\textwidth]{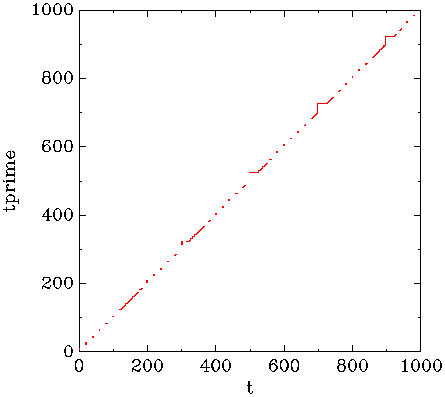}
		% NOTE: generated using ./imitator gear_1000.imi -mode EFunsafe -no-inclusion-test-in-EF -incl -verbose experiments -output-result -output-cart -output-cart-y-max 1000 -output-cart-x-max 1000 (by setting 't, tprime, p1' as order for the declared parameter in the model)

		\caption{Projection onto $t$ and~$t'$}
	\end{subfigure}

	\caption{Visualizing a large number of matches for \textsc{Gear} ($|\word = 1467|$)}
	\label{figure:projections:gear}
\end{figure*}
%%%%%%%%%%%%%%%%%%%%%%%%%%%%%%%%%%%%%%%%%%%%%%%%%%%%%%%%%%%%

%%%%%%%%%%%%%%%%%%%%%%%%%%%%%%%%%%%%%%%%%%%%%%%%%%%%%%%%%%%%
%%%%%%%%%%%%%%%%%%%%%%%%%%%%%%%%%%%%%%%%%%%%%%%%%%%%%%%%%%%%
\section{Conclusion}\label{section:conclusion}
%%%%%%%%%%%%%%%%%%%%%%%%%%%%%%%%%%%%%%%%%%%%%%%%%%%%%%%%%%%%
%%%%%%%%%%%%%%%%%%%%%%%%%%%%%%%%%%%%%%%%%%%%%%%%%%%%%%%%%%%%

We proposed a first approach to perform timed pattern matching in the presence of an uncertain specification.
This allows us to synthesize parameter valuations and intervals for which the specification holds on an input timed word.
Our implementation using \imitator{} may not completely allow for \emph{online} timed pattern matching yet, but already gives an interesting feedback in terms of parametric monitoring.
Our second algorithm aiming at finding minimal or maximal parameter valuations is less sensitive to state space explosion.
While our algorithms should be further optimized, we believe they pave the way for a more precise monitoring of real-time systems with an output richer than just timed intervals.

\paragraph*{Future works}
In~\cite{WHS17}, we proposed an approach that can apply for online timed pattern matching.
Its strength relied on the fact that an exhaustive search was not necessary, thanks to a \emph{skipping} mechanism.
Our next work will be to combine our current approach with the skipping mechanism of~\cite{WHS17}.

\todo{data structures}

Another challenge is the interpretation (and the visualization) of the results of parametric timed pattern matching.
While the result of $\PTPMopt$ is natural, the fully symbolic result of~$\PTPM$ remains a challenge to be interpreted; for example, the 125,250 matches for \textsc{Blowup} means the union of 125,250 polyhedra in 5~dimensions.
We give a possible way to visualize such results in \cref{figure:projections:gear} for \textsc{Gear} ($|\word| = 1,467$): in particular, observe in \cref{figure:projections:gear:t-param} that a single point exceeds~3, only a few exceed~2, while the wide majority remain in~$[0,1]$.
This helps to visualize how fast the gear is changed from~1 to~2, and at what timestamps.

Also, exploiting the polarity of parameters, as done in~\cite{ADMN11} or in lower-bound/upper-bound parametric timed automata~\cite{HRSV02}\todo{cite \cite{ALime17}}, may help to improve the efficiency of \PTPMopt{}.

Finally, a natural extension of our work is to study monitoring using more expressive logics such as~\cite{BKMZ15}.

\ea{not sure if I should cite \cite{HPU17}…?}

% %%%%%%%%%%%%%%%%%%%%%%%%%%%%%%%%%%%%%%%%%%%%%%%%%%%%%%%%%%%%
% %%%%%%%%%%%%%%%%%%%%%%%%%%%%%%%%%%%%%%%%%%%%%%%%%%%%%%%%%%%%
% \section*{Acknowledgements}
% %%%%%%%%%%%%%%%%%%%%%%%%%%%%%%%%%%%%%%%%%%%%%%%%%%%%%%%%%%%%
% %%%%%%%%%%%%%%%%%%%%%%%%%%%%%%%%%%%%%%%%%%%%%%%%%%%%%%%%%%%%
% XXXXX

% HACK to jump to new page only for the submitted version
% \ifdefined\VersionWithComments
% \else
% 	\newpage
% \fi

%%%%%%%%%%%%%%%%%%%%%%%%%%%%%%%%%%%%%%%%%%%%%%%%%%%%%%%%%%%%%
%%%%%%%%%%%%%%%%%%%%%%%%%%%%%%%%%%%%%%%%%%%%%%%%%%%%%%%%%%%%%
\ifdefined\VersionAuthor
	\bibliographystyle{alpha} % alpha % plain
	\newcommand{\IJFCS}{International Journal of Foundations of Computer Science}
	\newcommand{\JLAP}{Journal of Logic and Algebraic Programming}
	\newcommand{\LNCS}{LNCS}
	\newcommand{\STTT}{International Journal on Software Tools for Technology Transfer}
	\newcommand{\ToPNoC}{Transactions on Petri Nets and Other Models of Concurrency}
\else
	\bibliographystyle{IEEEtran} % abbrv
	\newcommand{\IJFCS}{IJFCS}
	\newcommand{\JLAP}{JLAP}
	\newcommand{\LNCS}{LNCS}
	\newcommand{\STTT}{STTT}
	\newcommand{\ToPNoC}{ToPNoC}
\fi
\bibliography{PTPM}
%%%%%%%%%%%%%%%%%%%%%%%%%%%%%%%%%%%%%%%%%%%%%%%%%%%%%%%%%%%%%
%%%%%%%%%%%%%%%%%%%%%%%%%%%%%%%%%%%%%%%%%%%%%%%%%%%%%%%%%%%%%

% \newpage

% BEGIN AUTHOR VERSION
\AuthorVersion{%

\appendix

%%%%%%%%%%%%%%%%%%%%%%%%%%%%%%%%%%%%%%%%%%%%%%%%%%%%%%%%%%%%
%%%%%%%%%%%%%%%%%%%%%%%%%%%%%%%%%%%%%%%%%%%%%%%%%%%%%%%%%%%%
\subsection{Handling thousands of polyhedral disjuncts}\label{appendix:polyhedra}
%%%%%%%%%%%%%%%%%%%%%%%%%%%%%%%%%%%%%%%%%%%%%%%%%%%%%%%%%%%%
%%%%%%%%%%%%%%%%%%%%%%%%%%%%%%%%%%%%%%%%%%%%%%%%%%%%%%%%%%%%

\imitator{} handles non-convex constraints (finite unions of polyhedra); while most case studies solved by \imitator{} in the past handle simple constraints (made of a few disjuncts), the experiments in this manuscript may handle up to \emph{dozens of thousands} of such polyhedra.
We therefore had to disable an inclusion test of a newly computed state into the already computed constraint: this test usually has a very interesting gain but, on our complex polyhedra, it had disastrous impact on the performance, due to the inclusion check of a (simple) new convex polyhedron into a disjunction of dozens of thousands of convex polyhedra.
To disable this check, we added a new option (not set by default) to the master branch of \imitator{}, and used it in all our experiments.

%%%%%%%%%%%%%%%%%%%%%%%%%%%%%%%%%%%%%%%%%%%%%%%%%%%%%%%%%%%%
%%%%%%%%%%%%%%%%%%%%%%%%%%%%%%%%%%%%%%%%%%%%%%%%%%%%%%%%%%%%
\subsection{Significantly decreasing \imitator{} computation time}\label{appendix:time}
%%%%%%%%%%%%%%%%%%%%%%%%%%%%%%%%%%%%%%%%%%%%%%%%%%%%%%%%%%%%
%%%%%%%%%%%%%%%%%%%%%%%%%%%%%%%%%%%%%%%%%%%%%%%%%%%%%%%%%%%%

In a preliminary version of our experiments, we obtained a computation time of up to 71 times slower than in the current version (\eg{} 1201.2 instead of 16.90 for \PTPM{} applied to \textsc{Accel} with an input word of size~25,137).
It turned out that a frequent manual invocation of the OCaml garbage collector was responsible for this much slower time.
Removing this invocation made our experiments much faster, while keeping the usual speed for the collection of benchmarks of \imitator{}.
Recall that the experiments in this paper are unusual as the size of the input model (in terms of number of locations) is an order of magnitude larger than \imitator{}'s historical benchmarks.
}
% END AUTHOR VERSION

\end{document}